\documentclass[a4paper,11pt]{article}
\usepackage{graphicx}%
\usepackage{multirow}%
\usepackage{amsmath,amssymb,amsfonts}%
\usepackage{amsthm}%
\usepackage{mathrsfs}%
\usepackage[title]{appendix}%
\usepackage{xcolor}%
\usepackage{textcomp}%
\usepackage{manyfoot}%
\usepackage{booktabs}%
\usepackage{algorithm}%
\usepackage{algorithmicx}%
\usepackage{algpseudocode}%
\usepackage{listings}%
\usepackage{hyperref}
\usepackage{braket}
\usepackage{tikz}
\usetikzlibrary{quantikz2}

\usepackage[square,numbers]{natbib}
\providecommand{\U}[1]{\protect\rule{.1in}{.1in}}
\newtheorem{theorem}{Theorem}[section]

\newtheorem{corollary}[theorem]{Corollary}

\newtheorem{definition}[theorem]{Definition}
\newtheorem{example}[theorem]{Example}
\newtheorem{lemma}[theorem]{Lemma}

\newtheorem{proposition}[theorem]{Proposition}

\newtheorem{remark}[theorem]{Remark}

\newcommand{\BIGOP}[1]{\mathop{\mathchoice{\raise-0.22em\hbox{\huge $#1$}} {\raise-0.05em\hbox{\Large $#1$}}{\hbox{\large $#1$}}{#1}}}

\makeatletter\@addtoreset{equation}{section}\makeatother
\makeatletter\@addtoreset{figure}{section}\makeatother
\makeatletter\@addtoreset{table}{section}\makeatother
\DeclareMathOperator{\tr}{tr}
\DeclareMathOperator{\diag}{diag}

\newcommand{\rD}{\mathsf{D}}
\newcommand{\rH}{\mathsf{H}}
\newcommand{\rI}{\mathsf{I}}

\newcommand{\rP}{\mathsf{P}}

\newcommand{\rR}{\mathsf{R}}

\newcommand{\rU}{\mathsf{U}}
\newcommand{\rV}{\mathsf{V}}

\newcommand{\rA}{\mathsf{A}}
\newcommand{\rB}{\mathsf{B}}
\newcommand{\rC}{\mathsf{C}}

\newcommand{\rX}{\mathsf{X}}

\newcommand{\rZ}{\mathsf{Z}}
\newcommand{\rT}{\mathsf{T}}
\newcommand{\rO}{\mathsf{O}}

\usepackage{xcolor}
\definecolor{customgreen}{HTML}{006400} % Define a custom color named 'customgreen'

\lstdefinestyle{python}{
    language=Python,
    basicstyle=\ttfamily\small,
    commentstyle=\color{green},
    keywordstyle=\color{blue},
    stringstyle=\color{red},
    frame=single,
    numbers=left,
    numberstyle=\tiny\color{gray},
    breaklines=true,
    showstringspaces=false,
    captionpos=b,
    xleftmargin=2em,
    xrightmargin=2em,
    aboveskip=1em,
    belowskip=1em,
    tabsize=4,
    morekeywords={as, assert, async, await, break, class, continue,
                  def, del, elif, except, exec, finally, from, global,
                  import, lambda, nonlocal, pass, print, raise, return,
                  try, with, yield, self, True, False, None}
}

\begin{document}

\title{A mathematical model for a universal digital quantum computer with an application to the Grover-Rudolph algorithm}
\author{Antonio Falc\'{o}\footnote{Corresponding author}$^{\, \, \,,1},$ Daniela Falc\'{o}--Pomares$^{2}$ and Hermann G. Matthies$%
^{3}$ \\
%EndAName
$^{1}$ Departamento de Matem\'aticas, F\'{\i}sica y Ciencias
Tecnol\'ogicas,\\
Universidad Cardenal Herrera-CEU, CEU Universities \\
San Bartolom\'e 55,
46115 Alfara del Patriarca (Valencia), Spain\\
e-mail: \texttt{afalco@uchceu.es}\\
$^{2}$ Grupo de Investigaci\'on Bisite \emph{Universidad de Salamanca}\\
Calle Espejo s/n, 37007 Salamanca (Spain)   \\
e-mail: \texttt{{dfp99@usal.es} }\\
$^{3}$ Institute of Scientific Computing, \\
Technische Universität Braunschweig\\
Universitätsplatz 2, 38106 Braunschweig, Germany.\\
e-mail: \texttt{h.matthies@tu-bs.de}}
\date{}

\maketitle

\begin{abstract}
In this work, we develop a novel mathematical framework for universal digital quantum computation using algebraic probability theory. We rigorously define quantum circuits as finite sequences of elementary quantum gates and establish their role in implementing unitary transformations. A key result demonstrates that every unitary matrix in \(\mathrm{U}(N)\) can be expressed as a product of elementary quantum gates, leading to the concept of a universal dictionary for quantum computation.  We apply this framework to the construction of quantum circuits that encode probability distributions, focusing on the Grover-Rudolph algorithm. By leveraging controlled quantum gates and rotation matrices, we design a quantum circuit that approximates a given probability density function. Numerical simulations, conducted using Qiskit, confirm the theoretical predictions and validate the effectiveness of our approach. These results provide a rigorous foundation for quantum circuit synthesis within an algebraic probability framework and offer new insights into the encoding of probability distributions in quantum algorithms. Potential applications include quantum machine learning, circuit optimization, and experimental implementations on real quantum hardware.
\end{abstract}



\section{Introduction.}\label{sec1}

Quantum computing promises revolutionary computational capabilities by harnessing the principles of quantum mechanics to tackle problems intractable for classical computers \cite{Shor1997,Grover1996}. Early breakthroughs such as Shor’s prime factorization algorithm and Grover’s search algorithm exemplify this potential. However, further progress requires a more rigorous mathematical foundation for quantum computation. Current formulations, which typically rely on state-vector evolution in Hilbert space, often lack a seamless integration with classical probability theory. This gap hampers our ability to fully understand quantum algorithms in probabilistic terms and motivates the need for a unified framework.

To address these limitations, we introduce a novel mathematical definition of a quantum computer 
grounded in \emph{algebraic probability theory} \cite{Gudder88,Meyer95,Parthasarathy1992}. 
This framework provides a unified probabilistic interpretation of quantum computation 
that naturally extends classical probability concepts into the quantum domain. 
Specifically, we define quantum computation as the action of the unitary group on 
the manifold of rank-one density matrices, derived from solutions to the Liouville--von Neumann 
differential equation. 

This definition offers a geometric perspective on computation, capturing both quantum superposition 
and probabilistic mixtures within a single formalism. Moreover, in our framework, and
\emph{elementary quantum gate} is rigorously defined as a fundamental unitary operation acting on a qubit. 
Such a precise characterization establishes a standardized basis for comparing and classifying 
quantum algorithms, laying the groundwork for a more principled analysis of quantum programs.

Our approach is general enough to encompass the standard gate-based circuit model of quantum computation. In the conventional circuit model, an algorithm is implemented as a finite sequence of discrete quantum gates acting on an initial state \cite{Nielsen_Chuang_2010,Kitaev2002}. This gate-based paradigm provides an intuitive operational blueprint for designing quantum algorithms and underlies most current quantum hardware implementations. Notably, it is the foundation for today’s Noisy Intermediate-Scale Quantum (NISQ) devices. An alternative paradigm, \textit{measurement-based quantum computing} (MBQC), takes a fundamentally different route by driving computation through adaptive measurements on highly entangled resource states \cite{Raussendorf2001}. MBQC decouples entanglement generation from the computation itself and is particularly promising for certain error-correction schemes and specialized architectures. While effective in their domains, both the circuit model and MBQC treat unitary evolution and measurement-based randomness in separate theoretical frameworks.

In contrast, our algebraic probability framework offers a unified and mathematically rigorous foundation that inherently incorporates both the coherent unitary dynamics of quantum circuits and the probabilistic nature of quantum measurements. By representing quantum states as density matrices and quantum operations as unitary actions on these states, our model seamlessly integrates state evolution with measurement outcomes within one probabilistic formalism. This unified perspective not only subsumes the conventional circuit model but also provides a natural setting for statistical analysis of quantum algorithms. Consequently, the framework enables more systematic comparisons between different quantum algorithms and supports the explicit construction of the unitary operators that implement them.

The key contribution of this work is a new theoretical framework for \emph{universal digital quantum computing} that bridges quantum mechanics and classical probability theory. Our algebraic approach offers a novel lens through which to examine quantum algorithms, yielding both conceptual clarity and practical criteria for algorithm design and evaluation. To illustrate the power and generality of the framework, we derive a generalized version of the Grover–Rudolph algorithm within our model. This generalized algorithm, originally inspired by \cite{GroverRudolph}, demonstrates how encoding a problem’s probability distribution into a unitary operator can be achieved explicitly in our formalism. The success of this construction not only extends a well-known quantum algorithm but also showcases the capability of our approach to analyze and generalize quantum algorithms in a principled way.

It is important to clarify the scope of our model within the landscape of quantum computing paradigms. Quantum computers can be broadly categorized as analog, digital, or fully error-corrected systems \cite{2019NAC}. The \textit{analog} approach exploits continuous physical interactions (e.g., adiabatic evolution or Hamiltonian simulations) to perform computations, whereas the \textit{digital} approach implements algorithms via discrete gate operations on qubits. (Current NISQ devices fall under the digital category, albeit without full error correction.) Our work focuses exclusively on the gate-based digital paradigm, which encompasses present NISQ processors as well as future fault-tolerant machines. Accordingly, we refer to our model as a \textit{digital quantum computer} to emphasize its foundation in discrete gate operations. This focus aligns our framework with the quantum computing architectures used in practice, while our theoretical advances provide deeper insight into the structure and analysis of such gate-based quantum algorithms.

The remainder of this paper is structured as follows: Section~\ref{ch_2} reviews the necessary mathematical background on finite-dimensional algebraic probability theory. Section~\ref{ch_3} introduces our proposed framework for a universal digital quantum computer, emphasizing the role of unitary group actions in quantum computation. Section~\ref{ch_4} applies this framework to derive the generalized Grover–Rudolph algorithm, including an explicit construction of the unitary operator for probability encoding. Section~\ref{Conclusions} concludes the paper by summarizing the key findings and outlining potential directions for future research.

\section{Basic Notions of Algebraic Probability}\label{ch_2}

This section introduces the fundamental notions and results in algebraic probability theory in a finite-dimensional setting.

\bigskip

Classical probability theory is traditionally formulated in terms of measure spaces. A classical probability space is defined as a triple $(\Omega, \mathcal{F}, P)$, where:
\begin{itemize}
    \item $\Omega$ is the sample space, representing all possible outcomes of a stochastic experiment.
    \item $\mathcal{F} \subseteq 2^{\Omega}$ is a $\sigma$-algebra of events.
    \item $P: \mathcal{F} \to [0,1]$ is a probability measure assigning probabilities to events.
\end{itemize}
A random variable $X: (\Omega, \mathcal{F}, P) \to (\mathbb{R}, \mathcal{B}(\mathbb{R}))$ induces a probability measure $P_X$ on $\mathbb{R}$, given by:
\begin{equation}
    P_X(B) = P(X^{-1}(B)), \quad B \in \mathcal{B}(\mathbb{R}).
\end{equation}
This measure is referred to as the \emph{distribution} of $X$. In the multivariate case, the joint distribution of an $n$-tuple $(X_1, \dots, X_n)$ of random variables is defined analogously.

\bigskip

On the other hand, quantum mechanics requires an extension of classical probability theory to incorporate non-commutative observables. In classical probability, observables are represented by real-valued random variables, which commute under multiplication, allowing them to be simultaneously measured with certainty. However, in quantum mechanics, observables correspond to self-adjoint operators on a Hilbert space, and these operators generally do not commute. This non-commutativity reflects the fundamental uncertainty in quantum measurements. To accommodate this structure, algebraic probability theory replaces classical probability measures with states in operator algebras, giving rise to the framework of algebraic probability spaces. This framework provides a mathematical setting where probability is defined over a non-commutative algebra rather than a classical sample space. 

\begin{definition}
An \emph{algebraic probability space} (or \emph{quantum probability space}) is a pair $(\mathcal{A}, \varphi)$, where:
\begin{itemize}
    \item $\mathcal{A}$ is a unital associative $*$-algebra over $\mathbb{C}$.
    \item $\varphi: \mathcal{A} \to \mathbb{C}$ is a \emph{state}, meaning it satisfies:
    \begin{enumerate}
        \item \emph{Linearity:} $\varphi(\lambda a + \mu b) = \lambda \varphi(a) + \mu \varphi(b)$ for all $a, b \in \mathcal{A}$ and $\lambda, \mu \in \mathbb{C}$.
        \item \emph{Positivity:} $\varphi(a^*a) \geq 0$ for all $a \in \mathcal{A}$.
        \item \emph{Normalization:} $\varphi(\mathbb{I}) = 1$, where $\mathbb{I}$ is the unit of $\mathcal{A}$.
    \end{enumerate}
\end{itemize}
\end{definition}

In this setting, a \emph{quantum random variable} corresponds to a self-adjoint element $X \in \mathcal{A}$. The state $\varphi$ provides an expectation value, analogous to classical probability:
\begin{equation}
    \mathbb{E}[X] = \varphi(X).
\end{equation}
A classical probability space $(\Omega, \mathcal{F}, P)$ can be embedded into an algebraic probability space $(\mathcal{A}, \varphi)$ by defining:
\begin{itemize}
    \item $\mathcal{A} = L^\infty(\Omega, \mathcal{F}, P)$, the algebra of bounded measurable functions on $\Omega$.
    \item The involution is pointwise complex conjugation: $f^*(\omega) = \overline{f(\omega)}$.
    \item The state $\varphi$ is given by expectation: $\varphi(f) = \int_{\Omega} f(\omega) dP(\omega)$.
\end{itemize}
Thus, classical probability emerges as a commutative subcase of the algebraic probability framework.

\subsection{Finite-Dimensional Example: The Matrix Algebra \(\mathbb{M}_N(\mathbb{C})\)}

Let \( N \in \mathbb{N} \) be given, and consider the finite-dimensional complex Hilbert space
$\mathbb{H}_N := \mathbb{C}^N.$ Along this paper, vectors in \(\mathbb{H}_N\) are represented in Dirac notation. For example, a vector \(\ket{\Psi} \in \mathbb{H}_N\) is written as
\[
\ket{\Psi} = \begin{bmatrix} \Psi_1 \\ \vdots \\ \Psi_N \end{bmatrix}.
\]
Its corresponding dual vector in \(\mathbb{H}_N^{\star}\) is identified with the conjugate transpose, namely
\[
\ket{\Psi}^{\star} := \bra{\Psi} = \begin{bmatrix} \overline{\Psi_1} & \cdots & \overline{\Psi_N} \end{bmatrix},
\]
where \(\overline{\Psi_i}\) denotes the complex conjugate of \(\Psi_i\). This notation allows the inner product of two vectors \(\ket{\Psi}, \ket{\Phi} \in \mathbb{H}_N\) to be expressed as
\[
\bra{\Psi}\ket{\Phi} = \sum_{i=1}^{N} \overline{\Psi_i}\Phi_i = \overline{\bra{\Phi}\ket{\Psi}},
\]
with the associated norm given by
\[
\|\Psi\| = \sqrt{\bra{\Psi}\ket{\Psi}}.
\]

\subsubsection{\(\mathbb{M}_N(\mathbb{C})\) as a Tensor Product Space}

In quantum mechanics, the outer product \(\ket{\Psi}\bra{\Phi}\) is an element of the tensor product space \(M_N(\mathbb{C}):=\mathbb{H}_N \widehat{\otimes} \mathbb{H}_N^{\star}\), which we denote by
\(
\ket{\Psi}\,\widehat{\otimes}\,\bra{\Phi} := \ket{\Psi}\bra{\Phi}.
\)
Explicitly,
\[
\ket{\Psi}\bra{\Phi} = \begin{bmatrix} \Psi_1 \\ \vdots \\ \Psi_N \end{bmatrix}
\begin{bmatrix} \overline{\Phi_1} & \cdots & \overline{\Phi_N} \end{bmatrix}
= \begin{bmatrix}
\Psi_1\overline{\Phi_1} & \cdots & \Psi_1\overline{\Phi_N} \\
\vdots & \ddots & \vdots \\
\Psi_N\overline{\Phi_1} & \cdots & \Psi_N\overline{\Phi_N}
\end{bmatrix}.
\]
Recall that
the \emph{trace} is a functional $\operatorname{tr}: \mathbb{M}_N(\mathbb{C}) \to \mathbb{C}$ given by:
\begin{equation}
    \operatorname{tr}(\rA) = \sum_{j=1}^{n} a_{jj},
\end{equation}
for any matrix $\rA = (a_{jk}) \in \mathbb{M}
_N(\mathbb{C})$. The trace can also be computed as:
\begin{equation}
    \operatorname{tr}(\rA) =\sum_{j=1}^{N} \bra{e_j} \ket{\rA e_j},
\end{equation}
where $\{\ket{e_1}, \dots, \ket{e_N}\}$ is any orthonormal basis of $\mathbb{H}_N$.

\bigskip

The tensor product \(\widehat{\otimes}\) is bilinear and satisfies the trace property
\begin{equation}\label{trace_property}
\tr\bigl(\ket{\Psi}\bra{\Phi}\bigr) = \sum_{i=1}^N \overline{\Phi_i}\Psi_i = \bra{\Phi}\ket{\Psi}.
\end{equation}
Note that rank-one matrices generate the matrix algebra \(\mathbb{M}_N(\mathbb{C})=\mathbb{H}_N \otimes \mathbb{H}_N^{\star}\), that is,
\[
  \mathbb{M}_N(\mathbb{C}) = \mathrm{span}\Bigl\{ \ket{\Psi}\bra{\Phi} : \ket{\Psi} \in \mathbb{H}_N,\; \bra{\Phi} \in \mathbb{H}_N^{\star} \Bigr\}.
\]
If \(\{\ket{\Psi_1},\dots,\ket{\Psi_N}\}\) is a basis for \(\mathbb{H}_N\), then \(\{\ket{\Psi_i}\bra{\Psi_j} : 1\le i,j\le N\}\) forms a basis for \(\mathbb{M}_N(\mathbb{C})\). Indeed, $\mathbb{M}_N(\mathbb{C})$ is a a unital associative $*$-algebra over $\mathbb{C}.$

\subsubsection{Hilbert-Schmidt Norm and the Spectral Theorem}

Before introducing the Hilbert-Schmidt norm, we present the following notation. Given a matrix
\[
\rA = \begin{bmatrix}
A_{1,1} & \cdots & A_{1,N} \\
\vdots & \ddots & \vdots \\
A_{N,1} & \cdots & A_{N,N}
\end{bmatrix} \in \mathbb{M}_N(\mathbb{C}),
\]
its conjugate transpose is defined by
\[
\rA^{\star} = \begin{bmatrix}
\overline{A_{1,1}} & \cdots & \overline{A_{N,1}} \\
\vdots & \ddots & \vdots \\
\overline{A_{1,N}} & \cdots & \overline{A_{N,N}}
\end{bmatrix} \in \mathbb{M}_N(\mathbb{C}).
\]
A matrix \(\rA \in \mathbb{M}_N(\mathbb{C})\) is said to be \emph{self-adjoint} (or \emph{Hermitian}) if \(\rA = \rA^{\star}\). 

\bigskip

The Hilbert-Schmidt norm on \(\mathbb{M}_N(\mathbb{C})\) is defined as
\[
\|\rA\|_{HS} = \sqrt{\tr\bigl(\rA^{\star}\rA\bigr)},
\]
and, as shown in \cite{Hackbusch2012}, \((\mathbb{M}_N(\mathbb{C}), \|\cdot\|_{HS})\) is a tensor Hilbert space with inner product
\[
(\rA, \rB)_{HS} = \tr\bigl(\rA^{\star}\rB\bigr) \quad \text{for all } \rA,\rB \in \mathbb{M}_N(\mathbb{C}).
\]
From \eqref{trace_property} and the definition of the Hilbert--Schmidt norm, it follows that for any vector \(\ket{\Psi}\in \mathbb{H}_N\),
\[
\|\ket{\Psi}\bra{\Psi}\|_{HS} \;=\; \sqrt{\langle \Psi \mid \Psi\rangle} \;=\; \|\Psi\|.
\]

Since $(\mathbb{M}_N(\mathbb{C}),\|\cdot\|_{HS})$ is a Hilbert space, the dual space
\[
\mathbb{M}_N(\mathbb{C})^{\star} = \{\,\varpi: \mathbb{M}_N(\mathbb{C}) \to \mathbb{C} \mid \varpi \text{ is complex linear} \,\}
\]
is isometrically isomorphic to \(\mathbb{M}_N(\mathbb{C})\) via the Hilbert-Schmidt inner product. In particular, for every \(\varpi \in \mathbb{M}_N(\mathbb{C})^{\star}\) there exists a unique \(\rho \in \mathbb{M}_N(\mathbb{C})\) such that
\[
\varpi(\rA) = (\rho, \rA)_{HS} \quad \text{for all } \rA \in \mathbb{M}_N(\mathbb{C}).
\]
We thus identify the linear functional \(\varpi\) with the matrix \(\rho \in \mathbb{M}_N(\mathbb{C})\).

\subsubsection{States on $\mathbb{M}_N(\mathbb{C})$}

To characterize the states on \(\mathbb{M}_N(\mathbb{C})\), we use the following result (see Theorem~4.6.2 in \cite{FranzPrivault2016}).

\begin{theorem}
  \label{thm:state_density_matrix}
Let $(\mathbb{M}_N(\mathbb{C}),\varphi)$ be an algebraic probability space. Then, there exists a unique matrix $\rho \in M_N(\mathbb{C})$ such that
  \begin{equation}
      \varphi(A) = \operatorname{tr}(\rho A), \text{ for all } A \in M_N(\mathbb{C}).
  \end{equation}
  The matrix $\rho,$ called the \emph{density matrix} of the state $\varphi,$ satisfies the following conditions:
  \begin{enumerate}
    \item[(a)] \(\rho\) is Hermitian (i.e., \(\rho = \rho^{\star}\));
    \item[(b)] All eigenvalues of \(\rho\) are nonnegative (i.e., \(\sigma(\rho) \subset \mathbb{R}_+\));
    \item[(c)] \(\tr(\rho) = 1\).
  \end{enumerate}
  \end{theorem}

We denote the set of all density matrices in $\mathbb{M}_N(\mathbb{C})$  by
\(
\mathcal{S}(\mathbb{M}_N(\mathbb{C})).
\)
In order to decompose $\mathcal{S}(\mathbb{M}_N(\mathbb{C})),$ we introduce the set of matrices with fixed rank. For \(r = 0,1,\dots, N\), let
\[
\mathcal{M}_r(\mathbb{M}_N(\mathbb{C})) := \{\,\rA \in \mathbb{M}_N(\mathbb{C}) \mid \mathrm{rank}\,\rA = r\,\}.
\]
An immediate consequence is that the set of rank-one matrices is given by
\[
\mathcal{M}_1(\mathbb{M}_N(\mathbb{C})) = \Bigl\{ \rA \in \mathbb{M}_N(\mathbb{C}) \,\Big|\, \rA = \lambda\,\ket{\Psi}\bra{\eta},\; \ket{\Psi},\ket{\eta}\in \mathbb{H}_N\setminus\{\ket{0_N}\},\; \lambda \in \mathbb{C}\setminus\{0\} \Bigr\},
\]
where \(\ket{0_N}\) denotes the zero vector in \(\mathbb{H}_N\),
In particular, for any \(\ket{\Psi} \in \mathbb{H}_N \setminus \{\ket{0_N}\}\) the rank-one matrix \(\ket{\Psi}\bra{\Psi}\) satisfies
\[
\|\ket{\Psi}\bra{\Psi}\|_{HS} = \sqrt{\bra{\Psi}\ket{\Psi}} = \|\Psi\|.
\]
Furthermore, \(\mathbb{M}_N(\mathbb{C})\) can be decomposed as the disjoint union
\[
\mathbb{M}_N(\mathbb{C}) = \bigcup_{r=0}^{N} \mathcal{M}_r(\mathbb{M}_N(\mathbb{C})).
\]
Since the zero matrix \(\rO_N\) is excluded as a state, we define, for each integer \(r\) with \(1 \le r \le N\), the subset of states of rank \(r\) by
\[
\mathcal{S}_r(\mathbb{M}_N(\mathbb{C})) 
:= \mathcal{S}(\mathbb{M}_N(\mathbb{C})) \;\cap\; \mathcal{M}_r(\mathbb{M}_N(\mathbb{C})).
\]
This construction partitions the set of all states according to their ranks, namely
\[
\mathcal{S}(\mathbb{M}_N(\mathbb{C})) 
= \bigcup_{r=1}^{N} \mathcal{S}_r(\mathbb{M}_N(\mathbb{C})).
\]

We now introduce the unitary group on \(\mathbb{M}_N(\mathbb{C})\):
\[
\mathrm{U}(N) := \{\rU \in \mathbb{M}_N(\mathbb{C}) \mid \rU^{\star}\rU = \rU\rU^{\star} = \rI_N\},
\]
where \(\rI_N\) denotes the \(N \times N\) identity matrix, that is, the matrix whose diagonal entries are all \(1\) and whose off-diagonal entries are all \(0\).
For a Hermitian matrix \(\rA\), its \emph{spectrum} is defined by
\[
\sigma(\rA) = \{\,\lambda \in \mathbb{C} \mid (\rA - \lambda \rI_N) \text{ is not invertible} \,\}.
\]
For any \(\lambda \in \sigma(\rA)\), the corresponding eigenspace is
\[
\ker(\rA-\lambda \rI_N) = \{\,\ket{u} \in \mathbb{H}_N \mid \rA\ket{u} = \lambda\ket{u}\,\}.
\]
If we write
\[
\rA = \begin{bmatrix} \ket{a_1} & \cdots & \ket{a_N} \end{bmatrix},
\]
where $\ket{a_i} \in \mathbb{H}_N$ for $1\le i \le N,$ we denote the set of its column vectors by
\[
\mathrm{col}\,\rA = \{\ket{a_1},\dots,\ket{a_N}\}.
\]

By employing the Singular Value Decomposition (SVD) (see \cite[Chapter 4]{Ipsen2009}), one may prove the following version of the spectral theorem for Hermitian matrices.

\begin{theorem}[Spectral Theorem]\label{SVD}
Let \(\rA \in \mathbb{M}_N(\mathbb{C})\) be a Hermitian matrix with spectrum \(\sigma(\rA) = \{\lambda_1, \ldots, \lambda_N\}\) and corresponding orthonormal eigenvectors \(\{\ket{u_1}, \ldots, \ket{u_N}\}\). Then, there exists a unitary matrix 
\[
\rU = \bigl[\ket{u_1}\, \ket{u_2}\, \cdots\, \ket{u_N}\bigr] \in \mathrm{U}(N)
\]
and a diagonal matrix \(\rD = \diag(\lambda_1, \ldots, \lambda_N) \in \mathbb{M}_N(\mathbb{C})\) such that
\[
\rA = \rU \rD \rU^{\star} = \sum_{i=1}^{N} \lambda_i \ket{u_i}\bra{u_i}.
\]
Furthermore, for \(0 \le r \le N\), the matrix \(\rA\) has rank \(r\) if and only if
\[
r = \sum_{\lambda \in \sigma(\rA) \setminus \{0\}} \mathrm{rank}(\rP_\lambda),
\]
where 
\[
\rP_\lambda = \sum_{\ket{u} \in \mathrm{col}\,\rU \cap \ker(\rA-\lambda \rI_N)} \ket{u}\bra{u}.
\]
Thus, \(\rA\) may be written as
\begin{equation}\label{eq:SVD1}
\rA = \sum_{\lambda \in \sigma(\rA)} \lambda\, \rP_\lambda.
\end{equation}
\end{theorem}

A consequence of the above theorem is that $\rA \in \mathcal{M}_1(\mathbb{M}_N(\mathbb{C}))$ is a Hermitian matrix if and only if $\rA = \lambda\, \ket{\Psi}\bra{\Psi}$ for some unit vector $\ket{\Psi} \in \mathbb{H}_N$ and $\lambda \in \mathbb{R}\setminus\{0\}.$ Clearly, \(\|A\|_{HS} = |\lambda|.\) 
In particular, the extremal points of \(\mathcal{S}(\mathbb{M}_N(\mathbb{C}))\) are given by
\[
\mathcal{S}_1(\mathbb{M}_N(\mathbb{C})) = \Bigl\{ \rho \in \mathcal{S}(\mathbb{M}_N(\mathbb{C})) \mid \rho = \ket{\Psi}\bra{\Psi},\; \|\Psi\| = 1 \Bigr\}.
\]
The elements of $\mathcal{S}_1(\mathbb{M}_N(\mathbb{C}))$ are called \emph{pure states}.

\subsubsection{Random Variables, Events, and Laws}

Next, we define a \emph{random variable} in this setting.

\begin{definition}
A matrix $\rA \in \mathbb{M}_N(\mathbb{C})$ is called a random variable in the algebraic probability space \((\mathbb{M}_N(\mathbb{C}), \rho)\) if $\rA$ is Hermitian, i.e., \(\rA = \rA^{\star}\).
\end{definition}

Since \(\rA\) is Hermitian, by Theorem~\ref{SVD} it has the spectral decomposition
\[
\rA = \sum_{\lambda \in \sigma(\rA)} \lambda\, \rP_\lambda,
\]
where
\[
\rP_\lambda = \sum_{\ket{u} \in \mathrm{col}\,\rU \cap \ker(\rA-\lambda \rI_N)} \ket{u}\bra{u}.
\]
This decomposition allows us to define the \emph{event} corresponding to the outcome \(\rA = x\).

\begin{definition}
Let \(\rA\) be a random variable in the algebraic probability space \((\mathbb{M}_N(\mathbb{C}), \rho)\). For a real number \(x\), define the event \(\{\rA = x\}\) as the random variable in the algebraic probability space \((\mathbb{M}_N(\mathbb{C}), \rho)\) given by
\[
\rP_{\{\rA = x\}} :=
\begin{cases}
\displaystyle \sum_{\ket{u} \in \mathrm{col}\,\rU \cap \ker(\rA-x \rI_N)} \ket{u}\bra{u}, & \text{if } x \in \sigma(\rA), \\[1ex]
\rO_N, & \text{otherwise.}
\end{cases}
\]
Similarly, the events \(\{\rA < x\}\) and \(\{\rA > x\}\) may be defined.
\end{definition}

\begin{remark}
The set of random variables \(\{\rP_{\{\rA = x\}}: x \in \sigma(\rA)\}\) in the algebraic probability space \((\mathbb{M}_N(\mathbb{C}), \rho)\) satisfy the orthogonality relation
\[
\rP_{\{\rA = x\}}\,\rP_{\{\rA = y\}} = \delta_{x,y}\,\rP_{\{\rA = x\}},
\]
analogous to the indicator functions in classical probability.
\end{remark}

\begin{definition}
The \emph{law of the random variable} \(\rA\) in the algebraic probability space \((\mathbb{M}_N(\mathbb{C}), \rho)\) is defined as the function
\[
\mathbb{P}_{\rho}(\rA = x) := (\rho, \rP_{\{\rA = x\}})_{HS} =
\begin{cases}
\displaystyle \sum_{\ket{u} \in \mathrm{col}\,\rU \cap \ker(\rA-x \rI_N)} \tr\bigl(\rho\,\ket{u}\bra{u}\bigr), & x \in \sigma(\rA), \\[1ex]
0, & \text{otherwise.}
\end{cases}
\]
It follows that
\[
\sum_{\lambda \in \sigma(\rA)} \mathbb{P}_{\rho}(\rA = \lambda) = 1.
\]
\end{definition}

\begin{remark}
In \cite{FranzPrivault2016} the law of a random variable $\rA$ in the algebraic probability space \((\mathbb{M}_N(\mathbb{C}), \rho)\) is represented by the measure
$$
\mathcal{L}_{\rho}(\rA) = \sum_{\lambda \in \sigma(\rA)} \tr(\rho \rP_{\{\rA = \lambda\}}) \delta_{\lambda}
=  \sum_{\lambda \in \sigma(\rA)} \langle\rho,\rP_{\{\rA = \lambda\}} \rangle_{HS} \delta_{\lambda}.
$$
Here $\delta_{\lambda}$ is a Dirac delta considered as a positive and finite Radon measure.
\end{remark}

To illustrate the concept of the law of a random variable in the algebraic probability framework, we now present two concrete examples. The first example considers a Bernoulli random variable, which provides an intuitive probabilistic interpretation of the outcomes of quantum measurement. The second example generalizes this notion to arbitrary random variables, demonstrating how the spectral decomposition of Hermitian matrices naturally extends classical probability distributions into the quantum domain.

\subsubsection{Example (Bernoulli random variable)}
Let \(\rA\) be a random variable in the algebraic probability space \((\mathbb{M}_N(\mathbb{C}), \rho)\), where  
\[
\rA = \ket{u}\bra{u} \in \mathcal{M}_1(\mathbb{M}_N(\mathbb{C}))
\]
and  
\[
\rho = \ket{\Psi}\bra{\Psi} \in \mathcal{S}_1(\mathbb{M}_N(\mathbb{C})),
\]
with \(\ket{\Psi}\) being a unit vector in \(\mathbb{H}_N\).  
Since the spectrum of \(\rA\) is \(\sigma(\rA) = \{0,1\}\), the corresponding random variables are given by:
\[
\rP_{\{\rA = 1\}} = \frac{\ket{u}\bra{u}}{\|\ket{u}\|^2}, \quad
\rP_{\{\rA = 0\}} = \rI_N - \frac{\ket{u}\bra{u}}{\|\ket{u}\|^2}.
\]
Here, \(\rP_{\{\rA = 1\}}\) represents the projection onto the subspace spanned by  \(\ket{u}\), while \(\rP_{\{\rA = 0\}}\) is the projection onto its orthogonal complement.
If \(\theta_{\Psi,u}\) denotes the angle between \(\ket{\Psi}\) and \(\ket{u}\), then
\[
\mathbb{P}_{\ket{\Psi}\bra{\Psi}}(\rA = 1) = \tr\!\Bigl(\ket{\Psi}\bra{\Psi}\,\frac{\ket{u}\bra{u}}{\|\ket{u}\|^2}\Bigr)
= \frac{|\bra{\Psi}\ket{u}|^2}{\|\ket{u}\|^2} = \cos^2\theta_{\Psi,u} = p,
\]
and
\[
\mathbb{P}_{\ket{\Psi}\bra{\Psi}}(\rA = 0) = 1 - p = \sin^2\theta_{\Psi,u}.
\]

\subsubsection{Example (General random variable)}
Let \(\rA\) be a random variable in the algebraic probability space \((\mathbb{M}_N(\mathbb{C}), \rho)\) with spectral decomposition
\(
\rA = \rU \rD \rU^{\star} = \sum_{\lambda \in \sigma(\rA)} \lambda\, \rP_\lambda,
\)
where for each \(\lambda\),
\[
\rP_\lambda = \sum_{\ket{u} \in \mathrm{col}\,\rU \cap \ker(\rA-\lambda \rI_N)} \ket{u}\bra{u}.
\]
Assume that \(\rho \in \mathcal{S}_r(\mathbb{M}_N(\mathbb{C}))\) for some \(1 \le r \le N\). Since \(\rho\) is Hermitian, Theorem~\ref{SVD} guarantees that there exists \(\rV \in \mathrm{U}(N)\) and a diagonal matrix $\Sigma$ such that
\(
\rho = \rV \Sigma \rV^{\star} = \sum_{p \in \sigma(\rho)} p\, \rP_{p},
\)
with
\[
\rP_{p} = \sum_{\ket{\Psi} \in \mathrm{col}\,\rV \cap \ker(\rho-p \rI_N)} \ket{\Psi}\bra{\Psi},
\]
and
\(
\sum_{p \in \sigma(\rho)\setminus\{0\}} p = 1.
\)
Then, for each \(\lambda \in \sigma(\rA)\), the law of the random variable \((\rA,\rho)\) is given by
\[
\mathbb{P}_{\rho}(\rA = \lambda)
= \sum_{\ket{u} \in \mathrm{col}\,\rU \cap \ker(\rA-\lambda \rI_N)}  \quad \sum_{p \in \sigma(\rho)} \quad\sum_{\ket{\Psi} \in \mathrm{col}\,\rV \cap \ker(\rho-p \rI_N)}  p\, |\bra{\Psi}\ket{u}|^2.
\]

An important property of the law of a random variable in the algebraic probability framework is its invariance under unitary transformations. This invariance ensures that the probabilistic structure of a quantum system remains unchanged when evolving under a unitary transformation, a fundamental principle in quantum mechanics. The following theorem formalizes this property, demonstrating that the law of a random variable remains invariant under conjugation by unitary operators.

\begin{theorem}
Let $\rA$ be a random variable in the algebraic probability space $(\mathbb{M}_N(\mathbb{C}), \rho).$ The for each $\rV \in \mathrm{U}(N)$ the law of the random variable $\rA$ in the algebraic probability space $(\mathbb{M}_N(\mathbb{C}),\rho)$ is the same as the law of the random variable $\rV\rA\rV^{\star}$ in the algebraic probability space $(\mathbb{M}_N(\mathbb{C}), \rV\rho\rV^{\star}).$ Furthermore, $\sigma(\rA) = \sigma(\rV\rA \rV^{\star})$ and $\mathbb{P}_{\rho}(\rA = x) = \mathbb{P}_{\rV\rho\rV^{\star}}(\rV\rA\rV^{\star} = x)$ for all $x \in \sigma(\rA).$
\end{theorem}
\begin{proof}
First at all, observe that for any $\rV \in \mathrm{U}(N)$ we have
\begin{align*}
\mathbb{P}_{\rho}(\rA = x) = \langle \rho,\rP_{\{\rA = x\}} \rangle_{HS} = \langle \rV\rho\rV^{\star},\rV \rP_{\{\rA = x\}} \rV^{\star} \rangle_{HS}.
\end{align*}
Thus, the theorem follows if we show that $\rV \rP_{\{\rA = x\}} \rV^{\star} = \rP_{\{\rV\rA\rV^{\star} = x\}}$ for all $x \in \sigma(\rA).$ To prove it, let $\lambda \in \sigma(\rA)$ and $\ket{u} \in \mathrm{col}\,\rU \cap \ker(\rA-\lambda \rI_N).$ Then $\rA\ket{u} = \lambda \ket{u}$ and $\rV\ket{u} \in \ker(\rV\rA\rV^{\star}-\lambda \rI_N).$ Thus, we have 
\begin{align*}
  \rV \rP_{\{\rA = x\}} \rV^{\star} & = \rV \left( 
  \sum_{\ket{u} \in \mathrm{col}\,\rU \cap \ker(\rA-\lambda \rI_N)} \ket{u}\bra{u} \right) \rV^{\star} \\
  & = \sum_{\ket{u} \in \mathrm{col}\,\rU \cap \ker(\rA-\lambda \rI_N)} \rV\ket{u}\bra{\rV\ket{u}} \\  
  & = \sum_{\ket{z} \in \mathrm{col}\,\rV\rU \cap \ker(\rV\rA \rV^* -\lambda \rI_N)} \ket{z}\bra{z}
  = \rP_{\{\rV\rA\rV^{\star} = x\}}.
\end{align*}
This completes the proof.
\end{proof}

Now, with these fundamentals established, we next develop our unified quantum computing framework.

\section{An Algebraic Probability Framework for Universal Digital Quantum Computers}\label{ch_3}

In this section, we present a mathematical model for a universal digital quantum computer within the framework of algebraic probability. The model comprises two key components:

\begin{enumerate}
    \item \textbf{An \(n\)-qubit Quantum Processing Unit (QPU):} This serves as the fundamental quantum computational resource.
    \item \textbf{A dynamic state evolution mechanism:} Analogous to a Turing machine's tape, here it is realized through the action of the unitary group \(\mathrm{U}(N)\) on the space of quantum states.
\end{enumerate}

Formally, a quantum computation is defined as a finite sequence of unitary operations:
\[
    \rU \;=\; \rU_{\ell}\,\rU_{\ell-1}\,\cdots\,\rU_1 \;\in\; \mathrm{U}(N),
\]
where each \(\rU_k\) corresponds to a quantum gate. In practical implementations, these gates are typically \emph{local} unitaries that act on a small number of qubits, usually one, two, or three. Within our algebraic probability framework, a quantum gate is formally defined as a unitary operation on a qubit that preserves the probabilistic structure of states. This definition provides a rigorous mathematical foundation for constructing and analyzing quantum circuits. 
We conclude this section by establishing the formal properties of quantum gates within the unitary group \(\mathrm{U}(N)\), thereby laying the groundwork for the design and classification of quantum algorithms.

\subsection{A mathematical model of a Quantum Processor Unit}

To define the model of a quantum computer, we first introduce the concept of a \emph{quantum processor unit (QPU)} as a random variable in the algebraic probability space $(\mathbb{M}_2(\mathbb{C}),\rho),$ where $\rho$ is a pure state. 

\begin{definition}
A \emph{single-qubit quantum processor unit} is defined as a random variable $\rA$ in the algebraic probability space $(\mathbb{M}_2(\mathbb{C}),\rho),$ where $\rho \in \mathcal{S}_1(\mathbb{M}_2(\mathbb{C}))$ is a pure state. Thus, there exists a unitary vector $\ket{\Psi}\in \mathbb{H}_2$ such that $\rho = \ket{\Psi}\bra{\Psi}.$
\end{definition}

To give a more detailed description of a single-qubit QPU, we begin with an orthonormal basis for \(\mathbb{H}_2\):
\[
\mathcal{B}_2 := \Bigl\{\ket{0} = \begin{bmatrix} 1 \\ 0 \end{bmatrix}, \quad \ket{1} = \begin{bmatrix} 0 \\ 1 \end{bmatrix}\Bigr\}.
\]
The spectral decomposition of $\rA$ implies that
\[  
\rA = \rU (\lambda_0 \ket{0}\bra{0} + \lambda_1 \ket{1}\bra{1}) \rU^{\star},
\]
for some $\rU \in \mathrm{U}(2)$ and $\lambda_0,\lambda_1 \in \mathbb{R}$ with $\lambda_0 \le \lambda_1.$ The corresponding events $\{\rA = \lambda_z\}$ are
$\rP_{\{\rA = \lambda_z\}} = \ket{z}\bra{z}$ for
$z \in \{0,1\}.$ The law of the random variable $\rA$ is then given by
\[
\mathbb{P}_{\rho}(\rA = \lambda_z) = \tr\bigl(\rho\,\ket{z}\bra{z}\bigr) = |\bra{z}\ket{\Psi}|^2  \text{ for } z \in \{0,1\}.
\]

\bigskip

To extend this model to a QPU with multiple qubits, we recall the definition and some useful properties of Kronecker product. 

\bigskip

Let \(\rA \in \mathbb{C}^{M \times N}\) and \(\rB \in \mathbb{C}^{M' \times N'}\). The \emph{Kronecker product} \(\rA \otimes \rB \) is a matrix defined by
\[
\rA \otimes \rB =
\begin{pmatrix}
A_{1,1}\rB & A_{1,2}\rB & \dots & A_{1,N}\rB \\
A_{2,1}\rB & A_{2,2}\rB & \dots & A_{2,N}\rB \\
\vdots & \vdots & \ddots & \vdots \\
A_{M,1}\rB & A_{M,2}\rB & \dots & A_{M,N}\rB \\
\end{pmatrix} \in \mathbb{C}^{MM' \times NN'}.
\]
Some standard properties of the Kronecker product include:
\begin{enumerate}
  \item \(\rA \otimes (\rB \otimes \rC) = (\rA \otimes \rB) \otimes \rC\),
  \item \((\rA + \rB) \otimes \rC = (\rA \otimes \rC) + (\rB \otimes \rC)\),
  \item \(\rA \rB \otimes \rC \rD = (\rA \otimes \rC)(\rB \otimes \rD)\),
  \item \((\rA \otimes \rB)^{-1} = \rA^{-1} \otimes \rB^{-1}\),
  \item \((\rA \otimes \rB)^{\star} = \rA^{\star} \otimes \rB^{\star}\), and
  \item \(\tr(\rA \otimes \rB) = \tr(\rA) \tr(\rB).\)
  \item Assume that $\rA \in \mathbb{M}_N(\mathbb{C})$ and $\rB \in \mathbb{M}_M(\mathbb{C})$ are diagonal matrices. Then $\rA \otimes \rB$ is also diagonal.
\end{enumerate}
When the dimension of the Hilbert space is \(N = M^q\) for some \(M \ge 2\) and \(q \ge 2\), we can identify the tensor product space \(\mathbb{H}_{M^q}\) with the \(q\)-fold tensor product
\[
\mathbb{H}_{M^q} = \underbrace{\mathbb{H}_{M} \otimes \mathbb{H}_{M} \otimes \cdots \otimes \mathbb{H}_{M}}_{q \text{ times}} = \mathbb{H}_M^{\otimes q}.
\]
If \(\{\ket{e_1},\ldots,\ket{e_M}\}\) is an orthonormal basis of \(\mathbb{H}_M\), then an orthonormal basis of \(\mathbb{H}_{M^q}\) is given by
\[
\bigl\{\ket{e_{i_1}} \otimes \cdots \otimes \ket{e_{i_q}} : 1 \le i_k \le M,\; 1 \le k \le q \bigr\}.
\]
Since \(\mathbb{H}_{M^q}\) and its dual \((\mathbb{H}_{M^q})^{\star}\) form a Hilbert space, the associated matrix algebra is
\[
\mathbb{M}_{M^q}(\mathbb{C}) = \mathbb{H}_{M^q} \,\widehat{\otimes}\, (\mathbb{H}_{M^q})^{\star} = \mathbb{H}_M^{\otimes q} \,\widehat{\otimes}\, (\mathbb{H}_M^{\star})^{\otimes q}.
\]
It can be shown that \(\mathbb{M}_{M^q}(\mathbb{C})= \mathbb{M}_{M}(\mathbb{C})^{\otimes q}\); in other words, the algebra of \(M^q \times M^q\) matrices is linearly isomorphic to the tensor product of \(q\) copies of the algebra of \(M \times M\) matrices. In particular,
\begin{align*}
(\ket{e_{i_1}} \otimes \cdots \otimes \ket{e_{i_q}})\,\widehat{\otimes}\,(\bra{e_{j_1}} \otimes \cdots \otimes \bra{e_{j_q}})
&:= (\ket{e_{i_1}} \otimes \cdots \otimes \ket{e_{i_q}})
(\bra{e_{j_1}} \otimes \cdots \otimes \bra{e_{j_q}})\\[1ex]
&= \ket{e_{i_1}}\bra{e_{j_1}} \otimes \cdots \otimes \ket{e_{i_q}}\bra{e_{j_q}}\\[1ex]
&= (\ket{e_{i_1}}\,\widehat{\otimes}\,\bra{e_{j_1}}) \otimes \cdots \otimes (\ket{e_{i_q}}\,\widehat{\otimes}\,\bra{e_{j_q}}).
\end{align*}

\bigskip

For the remainder of this paper, we assume that the Hilbert space $\mathbb{H}_N$ has dimension \(N = 2^n\) for some \(n \in \mathbb{N}\). 

\bigskip

Now, consider $\rA_1,\ldots, \rA_n$  $n$-single-qubit QPUs, that is, $n$-random variables in the algebraic probability space $(\mathbb{M}_2(\mathbb{C}),\rho).$ where $\rho = \ket{\Psi}\bra{\Psi} \in \mathcal{S}_1(\mathbb{M}_2(\mathbb{C}))$ for some unit vector $\ket{\Psi} \in \mathbb{H}_2.$ Assume that the spectral decomposition of each $\rA_k$ is given by
\begin{align}\label{QPU}
\rA_k = \rU_k (\lambda_{0}^{(k)} \ket{0}\bra{0} + \lambda_{1}^{(k)} \ket{1}\bra{1}) \rU_k^{\star},  
\end{align}
where $\rU_k \in \mathrm{U}(2)$ and $\lambda_{0}^{(k)},\lambda_{1}^{(k)} \in \mathbb{R}$ with $\lambda_{0}^{(k)} \le \lambda_{1}^{(k)}.$

\begin{definition}
an $n$-qubit QPU is defined by the Kronecker product of $n$-single qubit QPUs, that is, $\rA_1 \otimes \cdots \otimes \rA_n$ where  $\rA_1,\ldots, \rA_n$  are $n$-single-qubit QPUs. We remark that, given 
 $\rA_1,\ldots, \rA_n$ random variables in the algebraic probability space  $(\mathbb{M}_2(\mathbb{C}),\rho),$where $\rho = \ket{\Psi}\bra{\Psi} \in \mathcal{S}_1(\mathbb{M}_2(\mathbb{C}))$ for some unit vector $\ket{\Psi} \in \mathbb{H}_2,$ then $\rA_1 \otimes \cdots \otimes \rA_n$ is a random variable in the algebraic probability space $(\mathbb{M}_{2^n}(\mathbb{C}),\rho^{\otimes n}),$ where $\rho^{\otimes n} = (\ket{\Psi}\bra{\Psi})^{\otimes n} = \ket{\Psi^{\otimes n}}\bra{\Psi^{\otimes n}}  \in \mathcal{S}_1(\mathbb{M}_{2^n}(\mathbb{C})).$ 
\end{definition}

Let $\rA:= \rA_1 \otimes \cdots \otimes \rA_n$ be an $n$-qubit QPU. By using \eqref{QPU}, we have
\begin{align*}
\rA & = \bigotimes_{k=1}^n (\rU_k \rA_k \rU_k^{\star}) \\
& =\bigotimes_{k=1}^n \rU_k \bigotimes_{k=1}^n \left(\lambda_{0}^{(k)} \ket{0}\bra{0} + \lambda_{1}^{(k)} \ket{1}\bra{1}\right) \bigotimes_{k=1}^n \rU_k^{\star} \\ 
& = \bigotimes_{k=1}^n \rU_k \left(\sum_{z_1\in \{0,1\}} \cdots \sum_{z_n \in \{0,1\}}\lambda_{z_1}^{(1)} \cdots \lambda_{z_n}^{(n)} \ket{z_1}\bra{z_1} \otimes \cdots \otimes \ket{z_n}\bra{z_n} \right) \bigotimes_{k=1}^n \rU_k^{\star} \\
& = \bigotimes_{k=1}^n \rU_k \left(\sum_{z_1\in \{0,1\}} \cdots \sum_{z_n \in \{0,1\}}\lambda_{z_1}^{(1)} \cdots \lambda_{z_n}^{(n)} \ket{z_1\cdots z_n}\bra{z_1 \cdots z_{n}}  \right) \left(\bigotimes_{k=1}^n \rU_k\right)^{\star},
\end{align*}
where 
$$
\ket{z_1\cdots z_n}:=\ket{z_1} \otimes \cdots \otimes \ket{z_n} 
$$
for $z_1\cdots z_n\in \{0,1\}^n.$ 

\bigskip

The above representation of the $n$-qubits QPU $\rA,$ allows to introduce the following notation. Identifying \(\{0,1\}\) with \(\mathbb{Z}_2\), the set \(\{0,1\}^n\) corresponds to \(\mathbb{Z}_2^n\). Next, we consider the natural bijection
\[
b_n : \mathbb{Z}_{2^n} \to \mathbb{Z}_2^n,
\]
given by
\[
b_n(k) = b_n\left(\sum_{i=1}^{n} z_i\,2^{i-1}\right) = z_1 z_2 \cdots z_n.
\]
Thus, to each integer \(k \in \mathbb{Z}_{2^n}\) we can associate a basis vector in \(\mathbb{H}_{2^n}\) by
\[
\ket{b_n(k)} = \ket{z_1 z_2 \cdots z_n} = \ket{z_1} \otimes \ket{z_2} \otimes \cdots \otimes \ket{z_n}.
\]
The canonical basis of \(\mathbb{H}_{2^n}\) is then
\(
\mathcal{B}_N := \{\ket{b_n(k)} : k \in \mathbb{Z}_{2^n}\} = \{\ket{\mathbf{z}} : \mathbf{z} \in \mathbb{Z}_2^n\}.
\)

\bigskip

Now, we can represent $n$-qubits QPU as
\begin{align*}
  \rA = \rA_1 \otimes \cdots \otimes \rA_n = \bigotimes_{k=1}^n \rU_k \left(\sum_{k \in \mathbb{Z}_N} \lambda_{b_n(k)}  \ket{b_{n}(k)} \bra{b_n(k)} \right) \bigotimes_{k=1}^n \rU_k^{\star},
\end{align*}
where 
$\lambda_{b_n(k)} = \lambda_{z_1\cdots z_n}:=\lambda_{z_1}^{(1)} \cdots \lambda_{z_n}^{(n)}$ for each $k \in \mathbb{Z}_{2^n}.$ In consequence, the event \(\{\rA = \lambda_{b_n(k)}\}\) corresponds to the random variable \(\rP_{\{\rA = \lambda_{b_n(k)}\}} = \ket{b_n(k)}\bra{b_n(k)}\) in the algebraic probability space $(\mathbb{M}_{2^n}(\mathbb{C}),\rho^{\otimes n}),$ which is a projector onto the linear space generated by the basis vector \(\ket{b_n(k)}\). Moreover, the law of the random variable \(\rA\) in the algebraic probability space $(\mathbb{M}_{2^n}(\mathbb{C}),\rho^{\otimes n})$ is given by
\begin{align*}
\mathbb{P}_{\rho^{\otimes n}}(\rA = \lambda_{b_n(k)}) & = \tr\bigl(\rho^{\otimes n}\,\ket{b_n(k)}\bra{b_n(k)}\bigr)  \\ 
& = |\bra{b_n(k)}\ket{\Psi^{\otimes n}}|^2 = \prod_{j=1}^n |\bra{z_j}\ket{\Psi}|^2.
\end{align*}
This completes the extension form a single-qubit QPU to an $n$-qubit QPU. From now on, to simplify notation  we will write $\{\rA = \lambda_{b_n(k)}\}$ as $\{\rA = k \}$ for each $k \in \mathbb{Z}_{2^n}.$

\subsection{Dynamics of State Evolution in Digital Quantum Computers}\label{ch:dynamic}

To develop our model of a \emph{Universal Digital Quantum Computer (UDQC)}, we introduce a discrete unitary-evolution framework. Specifically, we define the map
\begin{equation}\label{eq:QPU}
\mathcal{N} : \mathrm{U}(N) \times \mathcal{S}(\mathbb{M}_N(\mathbb{C})) \,\to\, \mathcal{S}(\mathbb{M}_N(\mathbb{C})), 
\quad 
(\rU, \rho) \;\mapsto\; \mathcal{N}_\rho(\rU) \;:=\; \rU\,\rho\,\rU^{\star},
\end{equation}
which is derived from the following theorem.

\bigskip

\begin{theorem}[Liouville-von Neumann Equation]\label{Liouville-vonNeumann}
Let \(\rH \in \mathbb{M}_N(\mathbb{C})\) be Hermitian, and consider the Liouville--von Neumann equation on 
\(\bigl(\mathbb{M}_N(\mathbb{C}),\|\cdot\|_{HS}\bigr)\):
\begin{align}\label{dynamics}
\frac{d\rho}{dt} \;=\; -\,i\,[\rH,\,\rho], \quad 
\rho(0) \;=\; \rho_0, 
\end{align}
where \([\rH,\,\rho] = \rH\rho - \rho\rH.\)
Then there is a unique solution given by
\[
\rho(t) \;=\; e^{-\,i\,t\,\rH}\,\rho_0\,e^{\,i\,t\,\rH}.
\]
Moreover, if \(\rho_0\) belongs to the set \(\mathcal{S}_r\bigl(\mathbb{M}_N(\mathbb{C})\bigr)\)
for some \(1 \le r \le N\), then \(\rho(t)\) remains in \(\mathcal{S}_r\bigl(\mathbb{M}_N(\mathbb{C})\bigr)\)
for all real \(t\).
\end{theorem}

\begin{proof}
See Appendix~\ref{Appendix_1}.
\end{proof}

\smallskip

The solution \(\rho(t) = e^{-\,i\,t\,\rH}\,\rho_0\,e^{\,i\,t\,\rH}\) to \eqref{dynamics} evolves the quantum state according to the Liouville--von Neumann equation. This evolution is unitary and preserves the probabilistic structure of the state, making it a natural model for quantum dynamics.

\bigskip

In the context of complex \(N \times N\) matrices, a matrix \(\mathsf{G}\) is called 
\emph{skew-Hermitian} if it satisfies
\(
\mathsf{G}^{\star} \;=\; -\,\mathsf{G}.
\) A fundamental fact is that the matrix exponential of
every skew-Hermitian matrix \(\mathsf{G}\) is  a 
unitary matrix:
\(\rU \;=\; e^{\mathsf{G}}.\) Moreover, this correspondence is essentially \emph{onto}: every unitary matrix \(\rU \in \mathrm{U}(N)\) 
can be written in the form \(\rU = e^{\mathsf{G}}\) for some skew-Hermitian \(\mathsf{G}\). In other words,
\[
\mathrm{U}(N) \;=\; 
\bigl\{\, e^{\mathsf{G}} \,\bigl|\; \mathsf{G}^{\star} = -\,\mathsf{G} \bigr\}.
\]

\bigskip

We are now in a position to define the map \(\mathcal{N}\) precisely. Take \((\rU, \rho) \in \mathrm{U}(N) \times \mathcal{S}(\mathbb{M}_N(\mathbb{C}))\). Since \(\rU \in \mathrm{U}(N)\), there exists a skew-Hermitian matrix \(\mathsf{G}\) such that \(e^{\mathsf{G}} = \rU.\) Setting \(\rH = i\,\mathsf{G}\) makes \(\rH\) Hermitian. By Theorem~\ref{Liouville-vonNeumann}, the unique solution of \eqref{dynamics} with \(\rho(0) = \rho_0\) is 
\[
\rho(t) \;=\; e^{-\,i\,t\,\rH}\,\rho_0\,e^{\,i\,t\,\rH}.
\]
Evaluating at \(t = 1\), we obtain
\[
\rho(1) \;=\; e^{-\,i\,\rH}\,\rho_0\,e^{\,i\,\rH}
\;=\;
e^{\,\mathsf{G}}\,\rho_0\,e^{-\,\mathsf{G}}
\;=\;
\rU \,\rho_0\, \rU^{\star}.
\]
Hence we define 
\[
\mathcal{N}(\rU,\rho) \;:=\; \rho(1) 
\;=\; \rU \,\rho_0\, \rU^{\star},
\]
and it follows that \(\mathcal{N}\) is well-defined and preserves the rank of quantum states.

\begin{remark}
The map \(\mathcal{N}\) can be viewed as a discrete dynamical system, where the unitary group \(\mathrm{U}(N)\) acts on the manifold \(\mathcal{S}(\mathbb{M}_N(\mathbb{C}))\). Furthermore, each subset \(\mathcal{S}_r(\mathbb{M}_N(\mathbb{C}))\) is invariant under this action.
\end{remark}

Consider the \(n\)-qubit QPU given by \(\rA := \rA_1 \otimes \cdots \otimes \rA_n\). Each \(\rA_k\) for \(1 \le k \le n\) is treated as a random variable in the algebraic probability space \(\bigl(\mathbb{M}_2(\mathbf{C}), \rho_0\bigr)\), where \(\rho_0 = \ket{0}\bra{0}\). Consequently, the tensor product \(\rA_1 \otimes \cdots \otimes \rA_n\) is a random variable in the enlarged algebraic probability space \(\bigl(\mathbb{M}_{2^n}(\mathbf{C}), \rho_0^{\otimes n}\bigr)\). Here,
\[
\rho_0^{\otimes n} \;=\; \ket{0^{\otimes n}}\bra{0^{\otimes n}} 
\;=\; \ket{b_n(0)}\bra{b_n(0)}.
\]
For each \(\rU \in \mathrm{U}(2^n)\), the map
\[
\mathcal{N}_{\rho_0^{\otimes n}}(\rU) 
\;=\; \rU \,\rho_0^{\otimes n} \,\rU^{\star}
\]
produces a new pure state in \(\mathbb{M}_{2^n}(\mathbf{C})\). Hence, one may regard 
\(\rA_1 \otimes \cdots \otimes \rA_n\) as a random variable in the algebraic probability space 
\(\bigl(\mathbb{M}_{2^n}(\mathbf{C}), \mathcal{N}_{\rho_0^{\otimes n}}(\rU)\bigr)\).

\subsection{A Universal Digital Quantum Computer}

We now define an \emph{universal digital quantum computer} within the algebra \(\mathbb{M}_N(\mathbb{C})\).

\begin{definition}
    A \emph{universal digital quantum computer} in \(\mathbb{M}_N(\mathbb{C})\) is a triple \((\rA, \rho_0, \mathcal{N})\) where:
    \begin{enumerate}
      \item[(a)] \(\rA = \rA_1 \otimes \cdots \otimes \rA_n\) is an $n$-qubit QPU with \(\rA\) being a Hermitian matrix of the form
      \[
      \rA = \rU \left(\sum_{k \in \mathbb{Z}_N} \lambda_{b_n(k)}\, \ket{b_n(k)}\bra{b_n(k)}\right) \rU^{\star},
      \]
      for some unitary matrix  $\rU \in \mathrm{U}(N),$ 
      \item[(b)]\(\rho_0 = \ket{b_n(0)}\bra{b_n(0)}\) is the fixed initial pure state, and
      \item[(c)] \(\mathcal{N} : \mathrm{U}(N) \times \{\rho_0\} \to \mathcal{S}_1(\mathbb{M}_N(\mathbb{C}))\) is the map defined in \eqref{eq:QPU}.
    \end{enumerate}
    \end{definition}

We now describe a \emph{quantum computational procedure} in a digital quantum computer.

\bigskip

\noindent\textbf{Quantum Computational Procedure:}
Given a universal digital quantum computer \(\bigl(\rA, \rho_0, \mathcal{N}\bigr)\) in \(\mathbb{M}_N(\mathbb{C})\) 
and a unitary matrix \(\rU\), the quantum computational procedure is composed of two steps:

\begin{enumerate}
  \item[\textbf{(a)}] \emph{Quantum Circuit Process:}  
  A quantum state evolution applies \(\rU\) to the initial state \(\rho_0\) via
  \[
  \mathcal{N}_{\rho_0}(\rU) 
  \;=\; \rU\,\rho_0\,\rU^{\star} 
  \;=\; \rho,
  \]
  yielding the random variable \(\rA\) in the output algebraic probability space 
  \(\bigl(\mathbb{M}_N(\mathbb{C}), \rho\bigr)\).

  \item[\textbf{(b)}] \emph{Measurement Process:}  
  A measurement of \(\rA\) in the output space \(\bigl(\mathbb{M}_N(\mathbb{C}), \rho\bigr)\) reveals 
  the outcome \(k \in \mathbb{Z}_N\) with probability \(\mathbb{P}_{\rho}\bigl(\rA = k\bigr)\). 
  This is summarized by the set
  \[
  \bigl\{(k,\, \mathbb{P}_{\rho}(\rA = k)) \in \mathbb{Z}_N \times [0,1] 
  \;\bigl|\; k \in \mathbb{Z}_N \bigr\}.
  \]
\end{enumerate}

\subsubsection{Examples}

\textbf{Example 1.} Let \((\rA, \rho_0, \mathcal{N})\) in \(\mathbb{M}_N(\mathbb{C})\) be an universal digital quantum computer in \(\mathbb{M}_N(\mathbb{C})\). Consider \(\rU = \rI_N\), then
\[
\mathcal{N}_{\rho_0}(\rI_N) = \rho_0 = \ket{b_n(0)}\bra{b_n(0)}.
\]
The law of the random variable \(\rA\) in the output the algebraic probability space \((\mathbb{M}_N(\mathbb{C}), \rho_0)\) is
\[
\mathbb{P}_{\rho_0}(\rA = k) = |\bra{b_n(0)}\ket{b_n(k)}|^2 = \delta_{0,k},
\]
for all \(k \in \mathbb{Z}_N\).

\bigskip

\noindent \textbf{Example 2 (Unitary matrix for a specific task).}  
Let \((\rA, \rho_0, \mathcal{N})\) be an universal digital quantum computer in \(\mathbb{M}_{2^3}(\mathbb{C})\). Suppose we wish to identify the numbers in \(\mathbb{Z}_8 = \{0,1,2,3,4,5,6,7\}\) that are powers of two. In other words, we wish to find the integers \(k \in \mathbb{Z}_8\) such that \(k = 2^m\) for \(m = 0, 1, 2\). The pure state that encodes the solution is given by $\rho = \ket{\Psi}\bra{\Psi}$ with
\[
\ket{\Psi} = \frac{1}{\sqrt{3}} \Bigl(\ket{b_3(1)} + \ket{b_3(2)} + \ket{b_3(4)}\Bigr) = \frac{1}{\sqrt{3}}\Bigl(\ket{100} + \ket{010} + \ket{001}\Bigr).
\]
The measurement process provides the law of the random variable \(\rA\) in the output the algebraic probability space \((\mathbb{M}_N(\mathbb{C}), \rho)\) that yields
\[
\mathbb{P}_{\rho}(\rA = k) = \tr\!\Bigl(\ket{\Psi}\bra{\Psi}\ket{b_3(k)}\bra{b_3(k)}\Bigr) = \frac{1}{3}\Bigl(\delta_{k,1} + \delta_{k,2} + \delta_{k,4}\Bigr),
\]
for each \(k \in \mathbb{Z}_8\). Hence, the outcome is the set \(\{1,2,4\}\) with full probability, demonstrating the effectiveness of the constructed pure state. The challenge, then, is to design a unitary matrix (i.e., an oracle or quantum algorithm) \(\rU \in \mathrm{U}(N)\) that produces the state \(\rho = \mathcal{N}_{\rho_0}(\rU)\) starting from the initial state \(\rho_0 = \ket{b_3(0)}\bra{b_3(0)}\).

\bigskip

From the above example we learn that unitary matrices are the quantum equivalent of classical computational algorithms. In this context, we need to introduce a notion of elementary computational unit in our quantum computational framework. Similar to the definition of QPU as a cluster of single qubits joined by the help of the Kronecker product, we assume that a quantum elementary operation corresponds to the action over a single qubit, by means the map defined in \eqref{eq:QPU}.

\subsection{Elementary Quantum Gates and Quantum Circuits}

We now introduce elementary quantum gates as the fundamental building blocks of quantum computation and formally define quantum circuits, which represent oracles or quantum algorithms within \(\mathrm{U}(N)\).

\bigskip

To achieve this, we first recall that a function \( f : \mathrm{U}(2) \to \mathrm{U}(N) \) is a group homomorphism if it satisfies the condition  
\[
f(\rU \rV) = f(\rU) f(\rV) \quad \text{for all } \rU, \rV \in \mathrm{U}(2).
\]
A homomorphism is called a monomorphism if it is injective. We denote by \(\mathrm{Emb}(\mathrm{U}(2),\mathrm{U}(N))\) the set of such group monomorphisms.
Before proceeding, we formally define the fundamental building blocks of quantum circuits: the elementary quantum gates.

\begin{definition}
A matrix \(\rU \in \mathrm{U}(N)\) is called an \emph{elementary quantum gate} if there exists a pair \((\rV, \mathfrak{i}) \in \mathrm{U}(2) \times \mathrm{Emb}(\mathrm{U}(2),\mathrm{U}(N))\) such that \(\rU = \mathfrak{i}(\rV)\). The set of all elementary quantum gates in \(\mathrm{U}(N)\) is denoted by \(\mathrm{QG}(N)\).
\end{definition}

From the above definition, it follows that \(\mathrm{QG}(2) = \mathrm{U}(2)\). 
The set of elementary quantum gates can be viewed as a \emph{dictionary}. 
More generally, a subset \(\mathcal{D} \subset \mathrm{U}(N)\) is called a dictionary 
if it is nonempty and satisfies the property
\[
\mathcal{D}^{\star} = \{ \rU^{\star} : \rU \in \mathcal{D} \} = \mathcal{D}.
\]

\begin{remark}
It is straightforward to verify that if \(\rU \in \mathrm{QG}(N)\), then \(\rU^{\star} \in \mathrm{QG}(N)\); hence, \(\mathrm{QG}(N)\) forms a dictionary in \(\mathrm{U}(N)\).
\end{remark}

For a given dictionary \(\mathcal{D}\), we denote by \(\langle \mathcal{D} \rangle\) the group it generates:
\[
\langle \mathcal{D} \rangle = \bigcup_{k \geq 0} \mathcal{D}^{(k)}, \quad \text{where } \mathcal{D}^{(0)} = \{\rI_N\}, \quad \mathcal{D}^{(k)} = \{\rU_1 \rU_2 \cdots \rU_k : \rU_i \in \mathcal{D} \}.
\]
A dictionary \(\mathcal{D}\) is said to be \emph{universal} if \(\langle \mathcal{D} \rangle = \mathrm{U}(N)\).

\bigskip

In a quantum computer with \(n\) qubits, a wire represents the action of \(\mathrm{U}(2)\) on a particular quantum processing unit (QPU), namely \(\rA_j\) for \(1 \leq j \leq n\). Each wire can be represented by elements of \(\mathrm{Emb}(\mathrm{U}(2),\mathrm{U}(N))\). For \(1 \leq j \leq n\), we define the mapping
\[
w_j^{(n)} : \mathrm{U}(2) \to \mathrm{U}(2^n), \quad \rU \mapsto w_j^{(n)}(\rU) := \rI_2^{\otimes (j-1)} \otimes \rU \otimes \rI_2^{\otimes (n-j)}.
\]
The action of an elementary quantum gate \(w_j^{(n)}(\rU)\) on the initial state of a universal digital quantum computer in \(\mathbb{M}_N(\mathbb{C})\) is given by
\begin{align*}
\mathcal{N}_{\ket{b_n(0)}\bra{b_n(0)}}(w_j^{(n)}(\rU))
& = \ket{b_{j-1}(0)} \bra{b_{j-1}(0)} \otimes \rU \ket{0}\bra{0} \rU^{\star} \otimes \ket{b_{n-j}(0)} \bra{b_{n-j}(0)} \\ 
& = \ket{b_{j-1}(0)} \bra{b_{j-1}(0)} \otimes \mathcal{N}_{\ket{0}\bra{0}}(\rU) \otimes \ket{b_{n-j}(0)} \bra{b_{n-j}(0)}.
\end{align*}
Thus, this corresponds to the action of \(\rU \in \mathrm{U}(2)\) on the \(j\)th initial state. 

\bigskip

To conclude, we define the set
\[
\mathcal{W}_N := \{w_j^{(n)}(\rV) : \rV \in \mathrm{U}(2),\, 1 \leq j \leq n\} \subset \mathrm{QN}(N),
\]
which also forms a dictionary that generates the group
\[
\langle \mathcal{W}_N \rangle := \mathrm{U}(2)^{\otimes n}.
\]
Although \(\mathrm{U}(2)^{\otimes n} \subset \mathrm{U}(N)\), the set \(\mathcal{W}_N\) is not a universal dictionary. Since \(\mathrm{U}(2)^{\otimes n}\) is a subgroup of the group generated by all quantum gates, i.e.,
\[
\mathrm{U}(2)^{\otimes n} \subset \langle \mathrm{QG}(N) \rangle,
\]
a natural question arises: is \(\mathrm{QG}(N)\) itself universal in \(\mathrm{U}(N)\)? The next result provides a positive answer to this question.

\begin{theorem}\label{thm:universal_dictionary}
The set of quantum gates \(\mathrm{QG}(N)\) forms a universal dictionary for the unitary group \(\mathrm{U}(N)\); that is, for every \(\rU \in \mathrm{U}(N)\), there exists an integer \(m = \frac{N(N-1)}{2}\) such that \(\rU\) can be expressed as a product of \(m\) elementary quantum gates.
\end{theorem}

\begin{proof}
See Appendix~\ref{Appendix_A}.
\end{proof}

An immediate consequence of Theorem~\ref{thm:universal_dictionary} is that every unitary matrix in \(\mathrm{U}(N)\) can be represented as a finite sequence of elementary quantum gates. This observation naturally motivates the following definition.

\begin{definition}
A \emph{quantum circuit} of length \(\ell \geq 0\) for a unitary matrix \(\rU \in \mathrm{U}(N)\) is defined as follows:  
\(\ell = 0\) if and only if \(\rU = \rI_N\). Otherwise, a quantum circuit is a finite sequence of elementary quantum gates  
\(\{\rU_1, \ldots, \rU_\ell\} \subset \mathrm{QG}(N) \setminus \{\rI_N\}\), satisfying
\[
\rU = \rU_\ell \rU_{\ell-1} \cdots \rU_1,
\]
with the additional condition that \(\rU_{i+1} \rU_i \neq \rI_N\) for all \(1 \leq i < \ell\).
\end{definition}

\bigskip

Thus, if \(\rU \in \mathrm{QG}(N)\) has circuit length \(\ell \geq 1\), then for each \(1 \leq k \leq \ell\), there exists a pair \((\rV_k, \mathfrak{i}_k) \in \mathrm{U}(2) \times \mathrm{Emb}(\mathrm{U}(2),\mathrm{U}(N))\) such that \(\rU_k = \mathfrak{i}_k(\rV_k)\), and
\[
\rU = \mathfrak{i}_\ell(\rV_\ell) \, \mathfrak{i}_{\ell-1}(\rV_{\ell-1}) \cdots \mathfrak{i}_1(\rV_1).
\]
It is important to note that a given \(\rU \in \mathrm{U}(N)\) may admit representations as quantum circuits of different lengths, meaning that the representation is generally non-unique. 

In particular, the identity matrix \(\rI_N = \rI_2^{\otimes n}\) is itself an elementary quantum gate (e.g., \(w_1^{(n)}(\rI_2) = \rI_N\)) and serves as the empty word in the formal language where elementary matrices act as symbols in the alphabet \(\mathrm{U}(2) \times \mathrm{Emb}(\mathrm{U}(2),\mathrm{U}(N))\). Consequently, the circuit length \(\ell\) can be interpreted as the number of elementary one-qubit operations required to implement a given quantum circuit.

\begin{figure}[ht]
\centering
\includegraphics[scale=0.15]{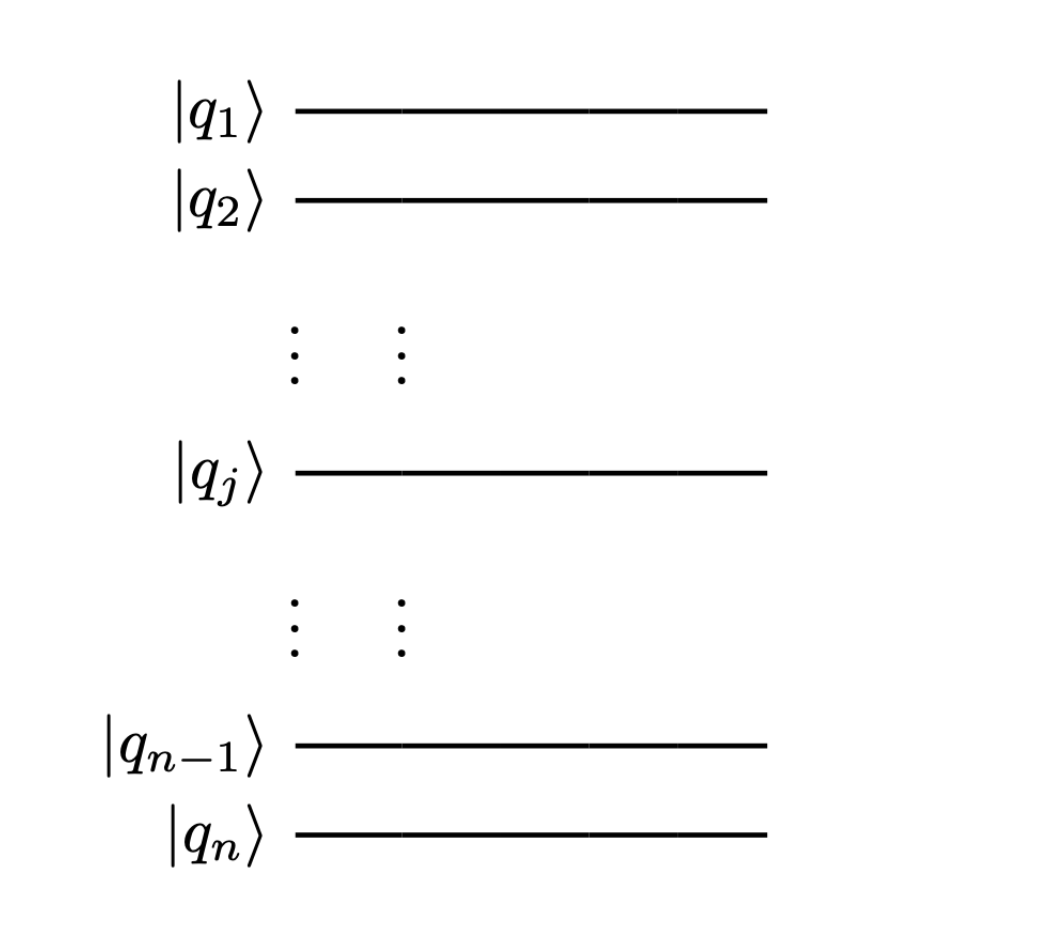}
\caption{Wire diagram for a \(n\)-qubit universal digital quantum computer.}
\label{circuit00}
\end{figure}

\begin{figure}[ht]
\centering
\includegraphics[scale=0.5]{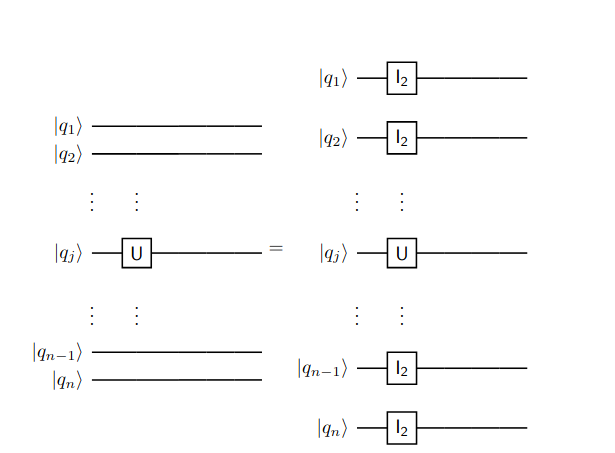}
\caption{Wire diagram of the elementary quantum gate \(w_j^{(n)}(\rU) = \rI_2^{\otimes (j-1)} \otimes \rU \otimes \rI_2^{\otimes (n-j)}.\)}
\label{circuit01}
\end{figure}

\begin{figure}[ht]
\centering
\includegraphics[scale=0.5]{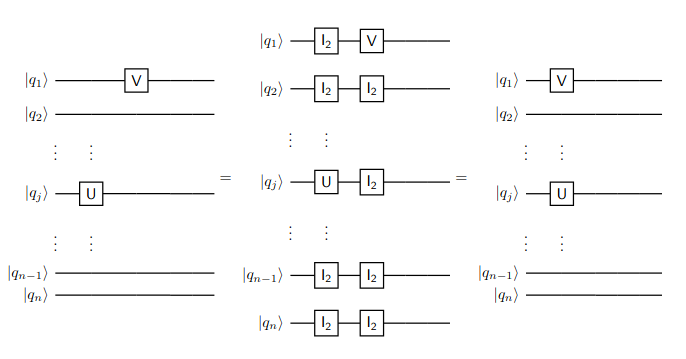}
\caption{Wire diagram of the product of elementary quantum gates \(w_1^{(n)}(\rV)w_j^{(n)}(\rU) = \rV \otimes \rI_2^{\otimes (j-2)} \otimes \rU \otimes \rI_2^{\otimes (n-j)}.\)}
\label{circuit02}
\end{figure}

A quantum circuit is typically depicted as a wire diagram consisting of:
\begin{enumerate}
\item[(a)] \emph{Wires}: Horizontal lines representing qubits (see Figure~\ref{circuit00}).
\item[(b)] \emph{Elementary quantum gates}: Symbols placed along the wires, each acting on a single qubit (see Figure~\ref{circuit01}).
\item[(c)] \emph{Directionality}: The circuit progresses from left to right, indicating the sequence of operations applied to the qubits (see Figure~\ref{circuit02}).
\end{enumerate}

\section{The Grover-Rudolph Algorithm}\label{ch_4}

The main goal of this section is to prove the following theorem.

\begin{theorem}[Grover-Rudolph]\label{theorem:Grover-Rudolph}
Let $(\rA, \rho_0, \mathcal{N})$ be a universal digital quantum computer in $\mathbb{M}_N(\mathbb{C})$. Given a non-negative function \(\varrho: [0,1] \to [0,\infty)\) satisfying the normalization condition
\[
\int_0^1 \varrho(x) \, dx = 1,
\]
there exists a quantum circuit \(\rU := \rU_{\ell} \rU_{\ell-1} \cdots \rU_1\) of length \(\ell = N-1\) such that the law of the random variable \(\rA\) in the outcome algebraic probability space \((\mathbb{M}_N(\mathbb{C}), \mathcal{N}_{\rho_0}(\rU))\) satisfies
\[
\mathbb{P}_{\mathcal{N}_{\rho_0}(\rU)}(\rA = k) = \int_{\frac{k}{2^n}}^{\frac{k+1}{2^n}} \varrho(x) \, dx, \quad \text{for each } k \in \mathbb{Z}_N.
\]
\end{theorem}

We proof the above theorem in two steps.

\subsection{A Trigonometric Decomposition of the Target Integral}

Our first goal is to express the integral
\[
I_k = \int_{\frac{k}{2^n}}^{\frac{k+1}{2^n}} \varrho(x)\,dx,
\]
for each \(k \in \mathbb{Z}_{2^n}\) in a form suitable for implementation through elementary quantum gates. To this end, we introduce the function
\[
\mathsf{T}:\mathbb{Z}_2 \times \mathbb{R} \to \mathbb{R}, \quad \mathsf{T}_z(x) = (\cos x)^{1-z}(\sin x)^z,
\]
which will serve as the building block in our decomposition.

\begin{theorem}\label{Rudolph-Grover-Proposition}
  Let \(\varrho: [0,1] \subset \mathbb{R} \to [0,\infty) \) be a nonnegative function satisfying the normalization condition
  \[
  \int _0^1 \varrho (x) \, dx = 1.
  \]
  Then, for each \(k \in \mathbb{Z}_{2^n}\) with binary representation \(b_n(k) = z_1z_2\cdots z_n\), 
  there exist parameters
  \[
  \theta_{z_2\cdots z_n},\,\theta_{z_3\cdots z_n},\,\ldots,\,\theta_{z_n}
  \]
  and a parameter \(\theta = \theta(\varrho)\) (independent of \(k\)) in the closed interval \(\left[0,\frac{\pi}{2}\right]\) such that
  \begin{align}
    \int_{\frac{k}{2^n}}^{ \frac{k+1}{2^n}}  \varrho(x) \, dx =  \mathsf{T}_{z_1}^2(\theta_{z_{2}\cdots z_n})\mathsf{T}_{z_{2}}^2(\theta_{z_{3}\cdots z_n}) \cdots \mathsf{T}_{z_{n-1}}^2(\theta_{z_n})\mathsf{T}_{z_n}^2(\theta). \label{eq:proposition}
  \end{align}
  Furthermore, the normalization condition is equivalent to
  \begin{align}\label{eq:proposition2}
    \sum_{z_1\cdots z_n \in \mathbb{Z}_2^n} \mathsf{T}_{z_1}^2(\theta_{z_{2}\cdots z_n})\mathsf{T}_{z_{2}}^2(\theta_{z_{3}\cdots z_n}) \cdots \mathsf{T}_{z_{n-1}}^2(\theta_{z_n})\mathsf{T}_{z_n}^2(\theta) = 1.
  \end{align}
\end{theorem}

  To prove Theorem~\ref{Rudolph-Grover-Proposition}, we partition the interval \([0,1]\) into equal subintervals \(2^\ell\) for each \(1\le \ell\le n\). For \(0\le k\le 2^\ell-1\), define
  \[
  z_{b_\ell(k)}^{(\ell)}:=\frac{k}{2^\ell}.
  \]
  We note that \( b_{\ell}(2^{\ell}) \) is not formally defined; however, for notational convenience, we assign \( z_{b_{\ell}(2^{\ell})}^{(\ell)} := 1 \).
  Thus, for fixed \(\ell\) we have
  \[
  0=z_{b_\ell(0)}^{(\ell)} < z_{b_\ell(1)}^{(\ell)} < \cdots < z_{b_\ell(2^\ell-1)}^{(\ell)} < z_{b_\ell(2^\ell)}^{(\ell)}=1.
  \]
  Next, we establish a series of technical lemmas (Lemmas~\ref{lema1} and \ref{lemma:decomposition}) to derive the recursive decomposition of the integral in terms of the functions $\rT_z$ and their corresponding angles.

  \begin{lemma}\label{lema1}
  Given $1 \le \ell \le n-1,$, for each $k \in \mathbb{Z}_{2^{\ell}}$ it holds $b_{\ell+1}(2k) = 0b_{\ell}(k)$ and $b_{\ell+1}(2k+1) = 1b_{\ell}(k).$ Furthermore,
  $b_{\ell+1}(2k+2) = 0b_{\ell}(k+1).$
  \end{lemma}

  \begin{proof}
    Assume that $k = z_12^0+z_22^2+\cdots+z_{\ell}2^{\ell-1}$, that is, $b_{\ell}(k)=z_1\cdots z_{\ell}.$ Then $2k =  z_12^1+z_22^2+\cdots+z_{\ell}2^{\ell}$ and hence
$b_{\ell+1}(2k) = 0z_1\cdots z_{\ell} = 0b_{\ell}(k).$ Using $2k +1 = 1\, 2^0+z_12^1+z_22^2+\cdots+z_{\ell}2^{\ell}$ we obtain  $b_{\ell+1}(2k+1) = 1b_{\ell}(k).$
    The last statement is straightforward.
    \end{proof}

  \bigskip

  Now,  given $1 \le \ell \le n-1,$ then for each $k \in \mathbb{Z}_{2^{\ell}}$ since
  $$
  z_{b_{\ell}(k)}^{(\ell)} = k \frac{1}{2^{\ell}} =  2k \frac{1}{2^{\ell+1}} =  z_{b_{\ell+1}(2k)}^{(\ell+1)},
  $$
  and
  $$
  z_{b_{\ell}(k+1)}^{(\ell)} = (k+1) \frac{1}{2^{\ell}} =  (2k+2) \frac{1}{2^{\ell+1}} =  z_{b_{\ell+1}(2k+2)}^{(\ell+1)},
  $$
  holds, we consider the following decomposition
  $$
  [z_{b_{\ell}(k)}^{(\ell)},z_{b_{\ell}(k+1)}^{(\ell)}] = [z_{b_{\ell+1}(2k)}^{(\ell+1)},z_{b_{\ell+1}(2k+1)}^{(\ell+1)}] \cup  [z_{b_{\ell+1}(2k+1)}^{(\ell+1)},z_{b_{\ell+1}(2k+2)}^{(\ell+1)}].
  $$
  From now one we will use the following notation
  $$
  I_{b_{\ell}(k)}^{(\ell)} := [z_{b_{\ell}(k)}^{(\ell)},z_{b_{\ell}(k+1)}^{(\ell)}].
  $$
  for $k \in \mathbb{Z}_{2^{\ell}}$ and $1 \le \ell \le n.$ Hence, from Lemma~\ref{lema1} we have the following corollary.
  
  \begin{corollary}\label{corolario1}
   Given $1 \le \ell \le n-1,$ then for each $k \in \mathbb{Z}_{2^{\ell}}$ it holds
  $$
  I_{b_{\ell}(k)}^{(\ell)} = I_{b_{\ell+1}(2k)}^{(\ell+1)} \cup I_{b_{\ell+1}(2k+1)}^{(\ell+1)} = I_{0b_{\ell}(k)}^{(\ell+1)} \cup I_{1b_{\ell}(k)}^{(\ell+1)}.
  $$
  \end{corollary}

  The above corollary allows us to write
  \begin{align}\label{interval:decomposition}
      \int_{I_{b_{\ell}(k)}^{(\ell)}} \varrho(x) dx =   \int_{{I_{0b_{\ell}(k)}^{(\ell+1)}}^{(\ell)}} \varrho(x) dx  +  \int_{I_{1b_{\ell}(k)}^{(\ell+1)}} \varrho(x) dx.
  \end{align}
  Recall that $\lfloor z \rfloor$ denotes the integer part function. Now, we have the following lemma.
  
  \begin{lemma}\label{lemma:decomposition}
  Let $k \in \mathbb{Z}_{2^n}$ be such that $b_n(k)=z_1z_2\cdots z_n.$ Then for $1 \le s \le n-1$ we have $$ b_n(k) = z_1 z_2 \cdots z_{s}  b_{n-s}(\lfloor 2^{-s}k \rfloor).
  $$
  Furthermore, $b_{n-s}(\lfloor 2^{-s}k \rfloor) = z_{s+1} \cdots z_n \in \mathbb{Z}_2^{n-s}.$
  \end{lemma}

\begin{proof}
Since $k = z_12^0+z_12^1+\cdots + z_n 2^{n-1},$ then $2^{-s}k =  z_12^{-s}+z_22^{-s+1}+\cdots +z_{s}2^{-1} + z_{s+1}2^0 +\cdots+z_n 2^{n-s},$ and
hence $\lfloor 2^{-s}k \rfloor = z_{s+1}2^0 +\cdots+z_n 2^{n-s},$ because $z_i \in \mathbb{Z}_2$ for all $1\leq i \leq n.$ Thus,
$b_n(k) = z_1z_2\cdots z_s b_{n-s}(\lfloor 2^{-s}k \rfloor)$ for all $1 \le s \le n-1.$ 
\end{proof}

To conclude this section, we present the proof of Theorem~\ref{Rudolph-Grover-Proposition}.
   
\begin{proof}[Proof of Theorem~\ref{Rudolph-Grover-Proposition}]
  Let $k \in \mathbb{Z}_{2^n}$ be such that $b_n(k)=z_1z_2\cdots z_n.$ Take $2 \le \ell \le n-1,$ from Lemma~\ref{lemma:decomposition}, choosing $s=n-\ell$, then $b_{\ell}(\lfloor 2^{-(n-\ell)}k \rfloor)=z_{n-\ell+1}\cdots z_n$ and we can write \eqref{interval:decomposition} as follows:
  \begin{align}\label{interval:decomposition1}
      \int_{I_{z_{n-\ell+1}\cdots z_n}^{(\ell)}} \varrho(x) dx =   \int_{I_{0z_{n-\ell+1}\cdots z_n}^{(\ell+1)}} \varrho(x) dx  +  \int_{I_{1z_{n-\ell+1}\cdots z_n}^{(\ell+1)}} \varrho(x) dx,
  \end{align}
  and if $ \int_{I_{z_{n-\ell+1}\cdots z_n}^{(\ell)}} \varrho(x) dx \neq 0,$ \eqref{interval:decomposition1} can be written as
  \begin{align}\label{cos-sin}
  \frac{\int_{I_{0z_{n-\ell+1}\cdots z_n}^{(\ell+1)}} \varrho(x) dx}{  \int_{I_{z_{n-\ell+1}\cdots z_n}^{(\ell)}} \varrho(x) dx} 
  + \frac{\int_{I_{1z_{n-\ell+1}\cdots z_n}^{(\ell+1)}} \varrho(x) dx}{\int_{I_{z_{n-\ell+1}\cdots z_n}^{(\ell)}} \varrho(x) dx} = 1.
  \end{align}
  Define
  \begin{align*}
  \theta_{z_{n-\ell+1}\cdots z_n} := 
  \left\{
    \begin{array}{ll}
    \arccos \sqrt{\frac{\int_{I_{0z_{n-\ell+1}\cdots z_n}^{(\ell+1)}} \varrho(x) dx}{ \int_{I_{z_{n-\ell+1}\cdots z_n}^{(\ell)}} \varrho(x) dx}} & \text{ if } \int_{I_{z_{n-\ell+1}\cdots z_n}^{(\ell)}} \varrho(x) dx \neq 0, \\
    \frac{\pi}{2} & \text{ if } \int_{I_{z_{n-\ell+1}\cdots z_n}^{(\ell)}} \varrho(x) dx = 0 \text{ and } z_{n-\ell+1} = 0, \\
    0 & \text{ if } \int_{I_{z_{n-\ell+1}\cdots z_n}^{(\ell)}} \varrho(x) dx = 0 \text{ and } z_{n-\ell+1} = 1.
\end{array}
\right.
\end{align*}
  and now, if $0 < \theta_{z_{n-\ell+1}\cdots z_n} < \frac{\pi}{2},$  from equation \eqref{cos-sin} we have 
  \begin{align*}
     \cos^2 \theta_{z_{n-\ell+1}\cdots z_n}  =  \frac{\int_{I_{0z_{n-\ell+1}\cdots z_n}^{(\ell+1)}} \varrho(x) dx}{  \int_{I_{z_{n-\ell+1}\cdots z_n}^{(\ell)}} \varrho(x) dx} \text{ and }  
     \sin^2 \theta_{z_{n-\ell+1}\cdots z_n} = \frac{\int_{I_{1z_{n-\ell+1}\cdots z_n}^{(\ell+1)}} \varrho(x) dx}{\int_{I_{z_{n-\ell+1}\cdots z_n}^{(\ell)}} \varrho(x) dx}.
  \end{align*}
In consequence, 
  $$
  \rT_z^2(\theta_{z_{n-\ell+1}\cdots z_n}) = \left\{
    \begin{array}{ll}
      \frac{\int_{I_{zz_{n-\ell+1}\cdots z_n}^{(\ell+1)}} \varrho(x) dx}{ \int_{I_{z_{n-\ell+1}\cdots z_n}^{(\ell)}} \varrho(x) dx} & \text{ if } \int_{I_{z_{n-\ell+1}\cdots z_n}^{(\ell)}} \varrho(x) dx \neq 0, \\
    0 & \text{ otherwise. }
\end{array}
\right.
  $$
  By taking $s = n- 1$ in Lemma~\ref{lemma:decomposition},  we have $b_{1}(\lfloor 2^{-(n-1)}k \rfloor) = z_n \in \mathbb{Z}_2,$ that is,
  $$
  \int_{I_{z_n}^{(1)}} \varrho(x) dx \in \left\{
  \int_{I_{0}^{(1)}} \varrho(x) dx, \int_{I_{1}^{(1)}} \varrho(x) dx,
  \right\}
  $$
  where it holds
  $$
  \int_{I_{0}^{(1)}} \varrho(x) dx + \int_{I_{1}^{(1)}} \varrho(x) dx = 1.
  $$
Define
  $$
  \theta = \theta(\varrho) := 
  \arccos \sqrt{\int_{I_{0}^{(1)}} \varrho(x) dx},
  $$
and hence 
$$
\cos^2 \theta = \int_{I_{0}^{(1)}} \varrho(x) dx \text{ and } \sin^2 \theta = \int_{I_{1}^{(1)}} \varrho(x) dx.
$$
Thus, $\rT_{z_n}^2(\theta) =  \int_{I_{z_n}^{(1)}}\varrho(x) dx.$

Recall that $I_{b_{n}(k)}^{(n)} =[z_{b_n(k)}^{(n)},z_{b_n(k+1)}^{(n)}] = [\frac{k}{2^n},\frac{k+1}{2^n}],$ and hence 
  $$
  \int_{\frac{k}{2^n}}^{ \frac{k+1}{2^n}} \varrho(x)dx = \int_{I_{b_{n}(k)}^{(n)}} \varrho(x) dx.
  $$
  Now, by using Lemma~\ref{lemma:decomposition}, we can write
    \begin{align*}
    \int_{I_{b_{n}(k)}^{(n)}} \varrho(x) dx  & = 
    \int_{I_{z_1\cdots z_n}^{(n)}} \varrho(x) dx  \\  
    & = 
    \frac{\int_{I_{z_1z_2\cdots z_n}^{(n)}} \varrho(x) dx}{ \int_{I_{z_2\cdots z_n}^{(n-1)}} \varrho(x) dx}
     \cdot \frac{\int_{I_{z_2z_3\cdots z_n}^{(n-1)}} \varrho(x) dx}{ \int_{I_{z_3\cdots z_n}^{(n-2)}} \varrho(x) dx} \cdots 
     \frac{\int_{I_{z_{n-1}z_n}^{(2)}} \varrho(x) dx}{ \int_{I_{z_n}^{(1)}} \varrho(x) dx} \cdot \int_{I_{z_n}^{(1)}} \varrho(x) dx \\
     & =  \rT_{z_1}^2(\theta_{z_{2}\cdots z_n})\rT_{z_{2}}^2(\theta_{z_{3}\cdots z_n}) \cdots \rT_{z_{n-1}}^2(\theta_{z_n})\rT_{z_n}^2(\theta),
    \end{align*}
    if 
    $$
    \int_{I_{z_{n-\ell+1}z_{n-\ell +2}\cdots z_n}^{(\ell)}} \varrho(x) dx \neq 0 \text{ holds for all } 1 \le \ell \le n.
    $$
Otherwise, if the above integral is zero for some $\ell,$ then 
$$
\rT_{z_{n-\ell +1}}(\theta_{z_{n-\ell+2}\cdots z_n}) = (\cos \theta_{z_{n-\ell+2}\cdots z_n})^{1-z_{n-\ell +1}}(\sin \theta_{z_{n-\ell+2}\cdots z_n})^{z_{n-\ell +1}}  = 0.
$$
Thus,
$$
\theta_{z_{n-\ell+2}\cdots z_n} = \left\{ 
\begin{array}{ll}
  \pi/2  & \text{if } z_{n-\ell +1} = 0, \\
 0 & \text{if } z_{n-\ell +1} = 1. 
\end{array}
\right.
$$

    This proves the first statement. The second one an is straightforward consequence of the fact
    $$
\int_{0}^1 \varrho(x)\, dx = \sum_{k \in \mathbb{Z}_{2^n}} \int_{I_{b_{n}(k)}^{(n)}} \varrho(x) dx = 1,
    $$
    and the proof is complete.
  \end{proof}

\subsection{Proof of Theorem~\ref{theorem:Grover-Rudolph}}  

Theorem~\ref{theorem:Grover-Rudolph} is a direct consequence of Theorem~\ref{Rudolph-Grover-Proposition} together with the following result.

  \begin{theorem}\label{theorem:Grover-RudolphX}
  Let $(\rA, \rho_0, \mathcal{N})$ be a universal digital quantum computer in $\mathbb{M}_N(\mathbb{C}).$ Then, there exists a quantum circuit $\rU$ of length $N-1$ such that the law of the random variable $\rA$ in the outcome algebraic probability space $(\mathbb{M}_N(\mathbb{C}), \mathcal{N}_{\rho_0}(\rU))$ satisfies the following condition for each $k \in \mathbb{Z}_{2^n}$, where $b_n(k) = z_1z_2\cdots z_n$:
  \begin{equation}
    \mathbb{P}_{\mathcal{N}_{\rho_0}(\rU)}(\rA = k)
    = \mathsf{T}_{z_1}^2(\theta_{z_{2}\cdots z_n})\mathsf{T}_{z_{2}}^2(\theta_{z_{3}\cdots z_n}) \cdots \mathsf{T}_{z_{n-1}}^2(\theta_{z_n})\mathsf{T}_{z_n}^2(\theta),
  \end{equation}
  where the parameters $\theta_{z_k \cdots z_n}$ for $2 \leq k \leq n$ and $\theta$ appear in the decomposition \eqref{eq:proposition}.
\end{theorem}

Next, we introduce a subset of elementary quantum gates called control quantum gates, that it will be useful to prove Theorem~\ref{theorem:Grover-Rudolph}.

\subsubsection{Control quantum gates}

Fix $1 \le \ell \le n$  and consider $\rU \in \mathrm{U}(2).$ 
We define a map 
$$
\rC^{(\ell)}: \mathbb{Z}_2^{n-1} \times \mathrm{U}(2) \rightarrow \mathbb{M}_N(\mathbb{C}), \quad (\mathbf{z},\rU) \mapsto \rC_{\mathbf{z}}^{(\ell)}\rU,
$$
where $\mathbf{z}=z_1\cdots z_{\ell -1}z_{\ell+1} \cdots z_{n} \in \mathbb{Z}_2^{n-1}$ as 
$$
\rC_{\mathbf{z}}^{(\ell)}\rU := \ket{z_1\cdots z_{\ell -1}} \bra{z_1\cdots z_{\ell -1}} \otimes \rU \otimes \ket{z_{\ell+1} \cdots z_{n}} \bra{z_{\ell+1} \cdots z_{n}}.
$$
This matrix acts over 
$\ket{\mathbf{u}} =  \ket{u_1 \cdots u_{\ell-1} u_{\ell} u_{\ell +1} \cdots u_{n}}$ for $\mathbf{u} \in \mathbb{Z}_2^n$
as follows:
$$
\rC_{\mathbf{z}}^{(\ell)}\rU \ket{\mathbf{u}} = \scalebox{0.9}{$\left\{
\begin{array}{ccc}
  w_{\ell}^{(n)}(\rU)\ket{\mathbf{u}} & \text{if }  u_1\cdots u_{\ell-1} u_{\ell+1} \cdots u_{n} = z_1\cdots z_{\ell-1} z_{\ell+1} \cdots z_{n}, \\
  0 &  \text{ otherwise.}
\end{array}
\right.$}
$$

For $\rU,\rV,\rI_2 \in \mathrm{U}(2),$ we have the following properties:
\begin{align}\label{P1}
(\rC_{\mathbf{z}}^{(\ell)}\rU)^{\star} = \rC_{\mathbf{z}}^{(\ell)}\rU^{\star}, 
\end{align}
\begin{align}\label{P2}
\rC_{\mathbf{z}}^{(\ell)}\rU \cdot \rC_{\mathbf{z}}^{(\ell)}\rV = \rC_{\mathbf{z}}^{(\ell)}\rU \rV,
\end{align}
and
\begin{align}\label{P3}
(\rC_{\mathbf{z}}^{(\ell)}\rU)^{\star} \cdot \rC_{\mathbf{z}}^{(\ell)}\rU =  \rC_{\mathbf{z}}^{(\ell)} \rI_2.
\end{align}
Observe that if $\mathbf{z}, \mathbf{z}^{\prime} \in \mathbb{Z}_2^{n-1}$ satisfies $\mathbf{z} \neq \mathbf{z}^{\prime},$
since $\bra{\mathbf{z}}\ket{\mathbf{z}^{\prime}} = \rO_N$ then 
\begin{align}\label{anhilator}
\rC_{\mathbf{z}}^{(\ell)}\rU \cdot \rC_{\mathbf{z}^{\prime}}^{(\ell)}
\rV = \rC_{\mathbf{z}^{\prime}}^{(\ell)}\rV \cdot \rC_{\mathbf{z}}^{(\ell)} \rU  = \rO_N,
\end{align}
holds.
Now, for $\mathbf{z} \in \mathbb{Z}_{2}^{n-1}$ we define the following map
$$
\rC \rC_{\mathbf{z}}^{(\ell)}:\mathrm{U}(2) \rightarrow \mathbb{M}_N(\mathbb{C}), \quad \rU \mapsto \rC \rC_{\mathbf{z}}^{(\ell)}\rU:= \rC_{\mathbf{z}}^{(\ell)} \rU + \sum_{\scalebox{0.6}{$\begin{matrix}  
  \mathbf{z}^{\prime}  \in \mathbb{Z}_2^{n-1} \\ \mathbf{z}^{\prime}  \neq  \mathbf{z}
  \end{matrix}$}
  } \rC_{\mathbf{z}^{\prime}}^{(\ell)} \rI_2.
$$
The above unitary matrix acts over 
$\ket{\mathbf{u}} =  \ket{u_1 \cdots u_{\ell-1} u_{\ell} u_{\ell +1} \cdots u_{n}}$ for $\mathbf{u} \in \mathbb{Z}_2^n$
as follows:

$$
\rC\rC_{\mathbf{z}}^{(\ell)}\rU \ket{\mathbf{u}} =\scalebox{0.9}{$\left\{
\begin{array}{ll}
  w_{\ell}^{(n)}(\rU)\ket{\mathbf{u}} & \text{if } u_1\cdots u_{\ell-1} u_{\ell+1} \cdots u_{n} = z_1\cdots z_{\ell-1} z_{\ell+1} \cdots z_{n}, \\
  \ket{\mathbf{u}} & \text{otherwise.}
\end{array}
\right.$}
$$
Then, we have the following lemma.

\begin{lemma}\label{lemma:control-gates}
For each $\mathbf{z} \in \mathbb{Z}_{2}^{n-1}$ and $1 \le \ell \le n,$ the map $\rC \rC_{\mathbf{z}}^{(\ell)} \in \mathrm{Emb}(\mathrm{U}(2),\mathrm{U}(N)).$ 
\end{lemma}

\begin{proof}
First at all, given $\rU,\rV \in \mathrm{U}(2),$ by using \eqref{P2} and \eqref{anhilator} , we have that
\begin{align}\label{P4}
(\rC \rC_{\mathbf{z}}^{(\ell)}\rU) \, (\rC \rC_{\mathbf{z}}^{(\ell)}\rV)
= \rC \rC_{\mathbf{z}}^{(\ell)}(\rU\rV).
\end{align}
From \eqref{P3} we have
\begin{align}\label{P5}
(\rC \rC_{\mathbf{z}}^{(\ell)}\rU)^{\star} = \rC_{\mathbf{z}}^{(\ell)} \rU^* + \sum_{\scalebox{0.6}{$\begin{matrix}  
  \mathbf{z}^{\prime}  \in \mathbb{Z}_2^{n-1} \\ \mathbf{z}^{\prime}  \neq  \mathbf{z}
  \end{matrix}$}
  } \rC_{\mathbf{z}^{\prime}}^{(\ell)} \rI_2 = \rC \rC_{\mathbf{z}}^{(\ell)}\rU^{\star}.
\end{align}
Now, by using \eqref{P5} and \eqref{P4} we have
$$
(\rC \rC_{\mathbf{z}}^{(\ell)}\rU)^{\star} \, \rC \rC_{\mathbf{z}}^{(\ell)}\rU = \rC \rC_{\mathbf{z}}^{(\ell)}(\rU^{\star}\rU) = \rC \rC_{\mathbf{z}}^{(\ell)}\rI_2 = \rI_N.
$$
hence $\rC \rC_{\mathbf{z}}^{(\ell)}\rU \in \mathrm{U}(N).$ Since 
$\rC \rC_{\mathbf{z}}^{(\ell)}\rU = \rI_N$ if and only if $\rU = \rI_2,$ then $\rC \rC_{\mathbf{z}}^{(\ell)} \in \mathrm{Emb}(\mathrm{U}(2),\mathrm{U}(N)).$ This ends the proof.    
\end{proof}

From the above lemma, for each $\mathbf{z} \in \mathbb{Z}_{2}^{n-1},$ $1 \le \ell \le n,$ and $\rU \in \mathrm{U}(2)$ the unitary matrix
$$
\rC \rC_{\mathbf{z}}^{(\ell)}\rU:= \rC_{\mathbf{z}}^{(\ell)} \rU + \sum_{\scalebox{0.6}{$\begin{matrix}  
\mathbf{z}^{\prime}  \in \mathbb{Z}_2^{n-1} \\ \mathbf{z}^{\prime}  \neq  \mathbf{z}
\end{matrix}$}
} \rC_{\mathbf{z}^{\prime}}^{(\ell)} \rI_2, 
$$
is an elementary quantum gate in $\mathrm{U}(N)$ that we called \emph{control quantum gate.} 

\begin{figure}[ht]
\centering
\includegraphics[scale=0.5]{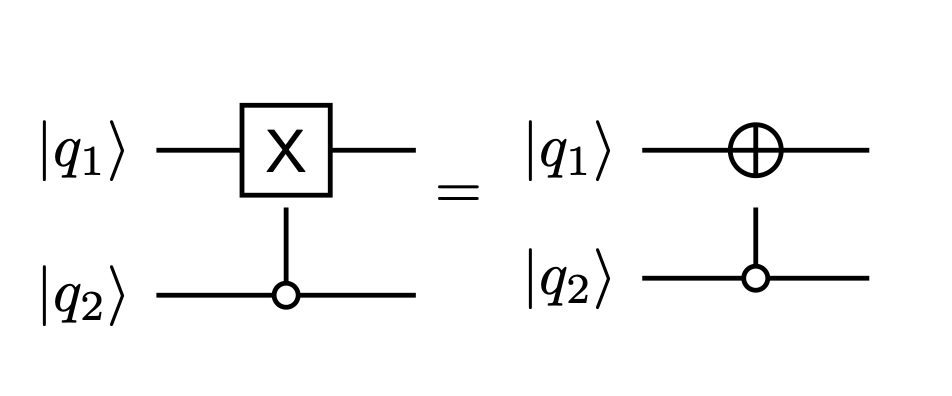}
\caption{Wire diagram for the CNOT elementary quantum gate $\rC\rC_{0}^{(1)}\rX.$}
\label{circuit03}
\end{figure}

\begin{example}
For $n=2,$ and $\ell = 1$ we have the following elementary quantum gate for each  $z \in \mathbb{Z}_2.$
Consider the unitary matrix $\rX=\begin{bmatrix} 0 & 1 \\ 1 & 0 \end{bmatrix}$ and the control quantum gate
$$
\rC\rC_{z}^{(1)}\rX = X \otimes \ket{z}\bra{z}+ \rI_2 \otimes \ket{1-z}\bra{1-z} \text{ for } z\in \mathbb{Z}_2.
$$
For $z=0,$ the unitary matrix $\rC\rC_{0}^{(1)}\rX$ it is known as the CNOT quantum gate (see Figure~\ref{circuit03}). 
\end{example}

The above example allows us to define for $2 \le \ell \le  n$ and $\mathbf{z} = z_{\ell} \cdots z_{n} \in \mathbb{Z}_{2}^{n-\ell+1}$ the control gate $\rI_{2^{\ell-1}} \otimes \rC\rC_{\mathbf{z}}^{(1)}\rU, $ where 
$$
\rC\rC_{\mathbf{z}}^{(1)}\rU = \rC_{\mathbf{z}}^{(1)} \rU + \sum_{\scalebox{0.6}{$\begin{matrix}  
\mathbf{z}^{\prime}  \in \mathbb{Z}_2^{n-\ell} \\ \mathbf{z}^{\prime}  \neq  \mathbf{z}
\end{matrix}$}} \rC_{\mathbf{z}^{\prime}}^{(1)} \rI_2,
$$
as follows. Consider the map 
$$
(\rI_{2^{\ell-1}} \otimes \rC\rC_{\mathbf{z}}^{(1)}):\mathrm{U}(2)
\longrightarrow \mathrm{U}(N), \quad \rU \mapsto (\rI_{2^{\ell-1}} \otimes \rC\rC_{\mathbf{z}}^{(1)})(\rU):=\rI_{2^{\ell-1}} \otimes \rC\rC_{\mathbf{z}}^{(1)}\rU.
$$
Then the next lemma it is not difficult to prove.

\begin{lemma}\label{lemma:control-gates1} Assume $2\le \ell \le n.$
    For each $\mathbf{z} \in \mathbb{Z}_{2}^{n-\ell+1}$ the map $\rI_{2^{\ell-1}} \otimes \rC\rC_{\mathbf{z}}^{(1)} \in \mathrm{Emb}(\mathrm{U}(2),\mathrm{U}(N)).$ 
    \end{lemma}
  
  Given $2 \le \ell \le n,$ from the above lemma, for each $\mathbf{z} \in \mathbb{Z}_{2}^{n-\ell +1}$ and $\rU \in \mathrm{U}(2)$ the unitary matrix
  $\rI_{2^{\ell-1}} \otimes \rC\rC_{\mathbf{z}}^{(1)}\rU$ is an elementary quantum gate in $\mathrm{U}(N).$
  
\bigskip

  Finally, for each $\alpha \in \mathbb{R}$ we will denote the rotation matrix  $\rR(\alpha) \in \mathrm{U}(2)$ as
$$
\rR(\alpha) := \begin{bmatrix}
\cos \alpha  & -\sin \alpha \\
\sin \alpha & \cos \alpha
\end{bmatrix} \in \mathrm{U}(2).
$$
Now, we have all the ingredients to start the proof of Theorem~\ref{theorem:Grover-RudolphX}.

\begin{proof}[Proof of Theorem~\ref{theorem:Grover-RudolphX}]
For each \(k \in \mathbb{Z}_{2^n}\) with binary representation 
\(
b_n(k) = z_1 z_2 \cdots z_n,
\)
the angles \(\theta_{z_j \cdots z_n}\) (for \(2 \le j \le n\)) and \(\theta\) are the parameters that appear in the decomposition~\eqref{eq:proposition} of Theorem~\ref{Rudolph-Grover-Proposition}.
To begin, we consider the elementary quantum gate
$\rU_1:=\rI_{2^{n-1}}\otimes \rR(\theta),$ 
then:
  \begin{align*}
    \ket{\psi_1} &= \rU_1 \ket{0}^{\otimes n} = ( \rI_{2^{n-1}}\otimes \rR(\theta))\ket{0}^{\otimes n} = \ket{0}^{\otimes n-1} \otimes (\cos \theta \ket{0} + \sin \theta \ket{1}) \\
    &= \ket{0}^{\otimes n-1} \otimes \displaystyle \sum_{z_n\in \mathbb{Z}_2} (\cos \theta)^{1-z_n}(\sin \theta )^{z_n}\ket{z_n} \\ 
    & = \ket{0}^{\otimes n-1} \otimes \displaystyle \sum_{z_n\in \mathbb{Z}_2} \rT_{z_n}(\theta)\ket{z_n}
 \end{align*}

Now, take the unitary matrix $\rU_2:= (\rI_{2^{n-2}}\otimes \rC\rC^{(1)}_{0}\rR(\theta_0)) (\rI_{2^{n-2}}\otimes \rC\rC^{(1)}_{1}\rR(\theta_1)),$ defined by the product of two elementary quantum gates. Then 
  \begin{align*}
    \ket{\psi_2} & = \rU_2\ket{\psi_1} = (\rI_{2^{n-2}}\otimes \rC\rC^{(1)}_{0}\rR(\theta_0)) (\rI_{2^{n-2}}\otimes \rC\rC^{(1)}_{1}\rR(\theta_1))\ket{\psi_1} \\
    & = (\rI_{2^{n-2}}\otimes \rC\rC^{(1)}_{0}\rR(\theta_0))  (\rI_{2^{n-2}}\otimes \rC\rC^{(1)}_{1}\rR(\theta_1))\left( \ket{0}^{\otimes n-2} \otimes \ket{0} \otimes  \displaystyle \sum_{z_n\in \mathbb{Z}_2} \rT_{z_n}(\theta)\ket{z_n}\right) \\
    &\\
    & = (\rI_{2^{n-2}}\otimes \rC\rC^{(1)}_{0}\rR(\theta_0))   \left( \ket{0}^{\otimes n-2} \otimes \ket{0} \otimes \displaystyle \rT_{0}(\theta)\ket{0} + \ket{0}^{\otimes n-2} \otimes R(\theta_1)\ket{0} \otimes \displaystyle \rT_{1}(\theta)\ket{1}\right)
    \\
    &\\
    & = \ket{0}^{\otimes n-2} \otimes \rR(\theta_0)\ket{0} \otimes \displaystyle \rT_{0}(\theta)\ket{0} + \ket{0}^{\otimes n-2} \otimes \rR(\theta_1)\ket{0} \otimes \displaystyle \rT_{1}(\theta)\ket{1}  \\
    &  \\
    & = \ket{0}^{\otimes n-2} \otimes \scalebox{0.87}{$\displaystyle \left( \cos \theta_0 \ket{0} \otimes \rT_0(\theta)\ket{0} + \sin \theta_0 \ket{1} \otimes \rT_0(\theta) \ket{0} + \cos \theta_1 \ket{0} \otimes  \rT_1(\theta) \ket{1} + \sin \theta_1 \ket{1}  \otimes \rT_{1}(\theta) \ket{1} \right)$} \\ 
    &  \\
    & = \ket{0}^{\otimes n-2} \otimes \scalebox{0.85}{$\displaystyle \left( \sum_{z_{n-1}\in \mathbb{Z}_2} (\cos \theta_{0})^{z_n}(\sin \theta_{0})^{1-z_n}\ket{z_n} \otimes \rT_0(\theta) \ket{0} + \sum_{z_{n-1}\in \mathbb{Z}_2} (\cos \theta_{1})^{z_n}(\sin \theta_{1})^{1-z_n}\ket{z_n}   \otimes \rT_{1}(\theta) \ket{1} \right)$} \\ 
    & \\
     &=\ket{0}^{\otimes n-2} \otimes \displaystyle \sum_{(z_{n-1}z_n)\in \mathbb{Z}_{2}^2}\rT_{z_{n-1}}(\theta_{z_n})\ket{z_{n-1}}\otimes \rT_{z_n}(\theta)\ket{z_n}.
 \end{align*}
Observe that we can write
$$
\rU_2 = \prod_{z_n \in \mathbb{Z}_2}(\rI_{2^{n-2}}\otimes \rC\rC^{(1)}_{z_n}\rR(\theta_{z_n}))  := (\rI_{2^{n-2}}\otimes \rC\rC^{(1)}_{0}\rR(\theta_0)) (\rI_{2^{n-2}}\otimes \rC\rC^{(1)}_{1}\rR(\theta_1)).
$$
For the next step we take
$$
\rU_3 = \prod_{(z_{n-1}z_n) \in \mathbb{Z}_{2}^2}(\rI_{2^{n-3}}\otimes \rC\rC^{(1)}_{z_{n-1}z_n}\rR(\theta_{z_{n-1}z_n})),
$$
which is the product of $2^2$ elementary quantum gates, to obtain
\begin{align*}
  \ket{\psi_3} & = \rU_3\ket{\psi_2} \\ 
  & = \prod_{(z_{n-1}z_n) \in \mathbb{Z}_{2}^2}\scalebox{0.9}{$(\rI_{2^{n-3}}\otimes \rC\rC^{(1)}_{z_{n-1}z_n}\rR(\theta_{z_{n-1}z_n})) \left( \ket{0}^{\otimes n-2} \otimes \displaystyle \sum_{(z_{n-1}z_n)\in \mathbb{Z}_{2}^2}\rT_{z_{n-1}}(\theta_{z_n})\ket{z_{n-1}}\otimes \rT_{z_n}(\theta)\ket{z_n}\right)$} \\
  & =   \ket{0}^{\otimes n-3} \otimes \scalebox{0.87}{$\displaystyle \sum_{z_{n-2}\in \mathbb{Z}_2} \sum_{(z_{n-1}z_n)\in \mathbb{Z}_{2}^2}(\cos \theta_{z_{n-1}z_n})^{z_{n-2}}
  (\sin \theta_{z_{n-1}z_n})^{1-z_{n-2}} \ket{z_{n-2}}  \otimes \rT_{z_{n-1}}(\theta_{z_n})\ket{z_{n-1}}\otimes \rT_{z_n}(\theta)\ket{z_n}$}\\
  & = \ket{0}^{\otimes n-3} \otimes \scalebox{0.9}{$ \displaystyle \sum_{(z_{n-2} z_{n-1} z_n)\in \mathbb{Z}_{2^3}}\rT_{z_{n-2}}(\theta_{z_{n-1}z_n})\ket{z_{n-2}}\otimes \rT_{z_{n-1}}(\theta_{z_n})\ket{z_{n-1}}\otimes \rT_{z_n}(\theta)\ket{z_n} $}\\
 &  = \ket{0}^{\otimes n-3} \otimes\scalebox{0.9}{$ \displaystyle \sum_{(z_{n-2} z_{n-1} z_n)\in \mathbb{Z}_{2^3}}\rT_{z_{n-2}}(\theta_{z_{n-1}z_n})\rT_{z_{n-1}}(\theta_{z_n})\rT_{z_n}(\theta)\ket{z_{n-2}z_{n-1}z_n}$}.
\end{align*}

Proceeding inductively, at the step $n$ we
take
$$
\rU_n = \prod_{(z_2 \cdots z_n) \in \mathbb{Z}_{2}^{n-1}}\rC\rC^{(1)}_{z_2 \cdots z_n}\rR(\theta_{z_2 \cdots z_n}),
$$
which is the product of $2^{n-1}$-elementary quantum gates, and then
\begin{align*}
    \ket{\psi_n} & = \rU_{n}\ket{\psi_{n-1}} \\
    & = \prod_{(z_2 \cdots z_n) \in \mathbb{Z}_{2}^{n-1}}\rC\rC^{(1)}_{z_2 \cdots z_n}\rR(\theta_{z_2 \cdots z_n}) \left( 
    \ket{0} + \sum_{(z_{2}\cdots z_n)\in \mathbb{Z}_{2}^{n-1}}\rT_{z_{2}}(\theta_{z_{3}\cdots z_n}) \cdots \rT_{z_n}(\theta)\ket{z_3 \cdots z_n}
    \right) \\
    & =  \displaystyle \sum_{z_1 \in \mathbb{Z}_2}\sum_{(z_{2}\cdots z_n)\in \mathbb{Z}_{2}^{n-1}} (\cos \theta_{z_2\cdots z_n})^{z_1} (\sin \theta_{z_2\cdots z_n})^{1-z_1} \ket{z_1} \otimes  \rT_{z_{2}}(\theta_{z_{3}\cdots z_n})\cdots \rT_{z_n}(\theta)\ket{z_2 \cdots z_n} \\
    & = \displaystyle  \sum_{(z_{1}\cdots z_n)\in \mathbb{Z}_{2}^{n}}\rT_{z_{1}}(\theta_{z_{2}\cdots z_n}) \cdots \rT_{z_n}(\theta)\ket{z_1 \cdots z_n}.
 \end{align*}
 Let $(\rA,\rho_0,\mathcal{N})$ be a universal digital quantum computer in $\mathbb{M}_N(\mathbb{C})$ and consider the quantum pure state $\ket{\psi_n}\bra{\psi_n}.$  Assume that for $k' \in \mathbb{Z}_{2^n}$ we have $b_n(k')=z_1'z_2'\cdots z_n'.$ Then the random variable $\rA$ has a law in the algebraic probability space $(\mathbb{M}_N(\mathbb{C}),\ket{\psi_n}\bra{\psi_n})$ given by
\begin{align*}
\mathbb{P}_{\ket{\psi_n}\bra{\psi_n}}(\rA = k') & = 
\tr\left(\ket{\psi_n}\bra{\psi_n}\ket{b_n(k')}\bra{b_n(k')}\right) \\
& = \tr\left(\ket{\psi_n}\bra{\psi_n}\ket{z_1'\cdots z_n'}\bra{z_1'\cdots z_n'}\right)  
\\
& =|\bra{\psi_n} \ket{z_1'\cdots z_n'}|^2 \\ 
& =\rT_{z_1'}^2(\theta_{z_{2}'\cdots z_n'})\rT_{z_{2'}}^2(\theta_{z_{3}'\cdots z_n'}) \cdots \rT_{z_{n-1}'}^2(\theta_{z_n'})\rT_{z_n'}^2(\theta).
\end{align*}

Since $\ket{\psi_n}=\rU_n \rU_{n-1} \cdots \rU_1 \ket{0}^{\otimes n} = \rU_n \rU_{n-1} \cdots \rU_1 \ket{b_n(0)},$ put $\rU:= \rU_n \rU_{n-1} \cdots \rU_1$ and hence 
$\ket{\psi_n}\bra{\psi_n} = \mathcal{N}_{\rho_0}(\rU).$ Thus,
$(\mathbb{M}_N(\mathbb{C}),\ket{\psi_n}\bra{\psi_n}) = (\mathbb{M}_N(\mathbb{C}), \mathcal{N}_{\rho_0}(\rU))$
is an output algebraic probability space obtained by means the unitary matrix $\rU.$ 
Now, for $2 \le \ell \le n,$  the unitary matrix
$$
\rU_{\ell} = \prod_{(z_{n-\ell+2} \cdots z_n) \in \mathbb{Z}_{2}^{\ell-1}}(\rI_{2^{n-\ell}}\otimes \rC\rC^{(1)}_{z_{n-\ell+2}\cdots z_n}\rR(\theta_{z_{n-\ell+2} \cdots z_n})),
$$
is the product of $2^{\ell}$-elementary quantum gates and $\rU_1 = \rI_{2^{n-1}}\otimes \rR(\theta)$ is a single elementary quantum gate. In consequence, the quantum circuit $\rU=\rU_{n-1} \cdots \rU_2\rU_1$ is the product of $2^n-1$- elementary quantum gates and hence, it has length $N-1.$ This ends the proof.
\end{proof}

\subsection{Constructing a Quantum Circuit for a Specific Task}

We now apply Theorem~\ref{theorem:Grover-Rudolph} to construct a quantum circuit that accomplishes the task outlined in Example 2 of Section~\ref{ch_3}. The objective is to construct a quantum circuit, \(\rU \in \mathrm{U}(8)\), such that
\[
\ket{\Psi} = \frac{1}{\sqrt{3}} \Bigl(\ket{b_3(1)} + \ket{b_3(2)} + \ket{b_3(4)}\Bigr) = \frac{1}{\sqrt{3}}\Bigl(\ket{100} + \ket{010} + \ket{001}\Bigr) = \rU \ket{b_3(0)}.
\]
where the measurement process determines the probability distribution of the random variable \(\rA\) in the algebraic probability space \((\mathbb{M}_N(\mathbb{C}), \rho)\), given by
\[
\mathbb{P}_{\rho}(\rA = k) = \frac{1}{3}\Bigl(\delta_{k,1} + \delta_{k,2} + \delta_{k,4}\Bigr).
\]

This probability can be rewritten as
\[
\mathbb{P}_{\rho}(\rA = k) = \int_{k/8}^{(k+1)/8}  \varrho(x) \, dx,
\]
where \(\varrho(x)\) is defined in terms of the characteristic function  \(\mathbf{1}_{I}(x)\), which is given by:
\[
\mathbf{1}_{I}(x) =
\begin{cases}
1, & x \in I, \\
0, & x \notin I.
\end{cases}
\]
Thus, we define
\[
\varrho(x):= \frac{8}{3} \left(
 \mathbf{1}_{[1/8,2/8]}(x) + \mathbf{1}_{[2/8,3/8]}(x) + \mathbf{1}_{[4/8,5/8]}(x)\right).
\]
For this case,
\[
\theta = \arccos \sqrt{\int_{0}^{1/2} \varrho(x) dx} =  \arccos \sqrt{\frac{2}{3}} = \frac{\pi}{6}.
\]

To construct \(\rU_2\), we use:
\[
\rU_2:= \prod_{z_2\in\mathbb{Z}_2}(\rI_{2}\otimes \rC\rC^{(1)}_{z_2}\mathsf{R}(\theta_{z_2}))=(\rI_{2}\otimes \rC\rC^{(1)}_{0}\mathsf{R}(\theta_0)) (\rI_{2}\otimes \rC\rC^{(1)}_{1}\mathsf{R}(\theta_1)).
\]
where
\[
\theta_0 = \arccos\sqrt{\frac{\int_{0}^{1/4}\varrho(x) \, dx}{\int_{0}^{1/2}\varrho(x) \, dx}} =  \frac{\pi}{2}, \quad \theta_1 = \arccos\sqrt{\frac{\int_{1/2}^{3/4} \varrho(x) \, dx}{\int_{1/2}^{1} \varrho(x) \, dx}} = 0.
\]
Since \(\rR(0)=\rI_2\), it follows that \( (\rI_{2}\otimes \rC\rC^{(1)}_{1}\mathsf{R}(\theta_1)) = \rI_{8}. \) Thus,
\[
\rU_2 = (\rI_{2}\otimes \rC\rC^{(1)}_{0}\mathsf{R}(\pi/2)).
\]
To complete the construction, we determine
\[
\mathsf{U}_3 = \prod_{(z_{1},z_2) \in \mathbb{Z}_{2}^2}\rC\rC^{(1)}_{z_{1}z_2}\mathsf{R}(\theta_{z_{1}z_2}),
\]
where
\[
\theta_{00} = \arccos\sqrt{\frac{\int_{0}^{1/8}\varrho(x) \, dx}{\int_{0}^{1/4}\varrho(x) \, dx}} = \frac{\pi}{2}, \quad \theta_{11} =\arccos\sqrt{\frac{\int_{3/4}^{7/8}\varrho(x) \, dx}{\int_{3/4}^{1}\varrho(x) \, dx}}=\frac{\pi}{2}.
\]
\[
\theta_{01}= \arccos \sqrt{\frac{\int_{1/2}^{5/8} \varrho(x) \, dx}{\int_{1/2}^{3/4} \varrho(x) \, dx}}= 0, \quad
\theta_{10} = \arccos\sqrt{\frac{ \int_{1/4}^{3/8} \varrho(x) \, dx}{\int_{1/4}^{1/2} \varrho(x) \, dx}}=0.
\]
Thus,
\[
\rU_3 = \rC\rC^{(1)}_{00}\mathsf{R}(\pi/2)\,\rC\rC^{(1)}_{11}\mathsf{R}(\pi/2).
\]
Consequently,
\[
\rU =  \rC\rC^{(1)}_{00}\mathsf{R}(\pi/2)\,\rC\rC^{(1)}_{11}\mathsf{R}(\pi/2) \, (\rI_{2}\otimes \rC\rC^{(1)}_{0}\mathsf{R}(\pi/2))\, (\rI_{4} \otimes \rR(\pi/6))
\]
is a quantum circuit of length \(4.\)

\bigskip

We now illustrate Theorem~\ref{theorem:Grover-Rudolph} with a specific numerical example.

\subsection{A Numerical Example}

Consider the following density function (see Figure~\ref{fig:Density_Function}):
\[
\varrho(x) =
\begin{cases}
4x, & 0 \leq x \leq \frac{1}{2}, \\
4 - 4x, & \frac{1}{2} \leq x \leq 1.
\end{cases}
\]

\begin{figure}[h]
    \centering
    \includegraphics[scale=0.5]{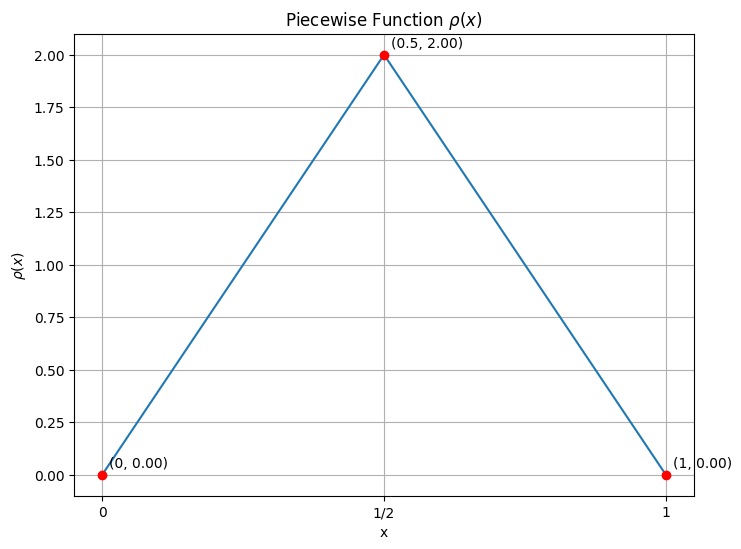}
    \caption{The density function \(\varrho(x)\).}
    \label{fig:Density_Function}
\end{figure}

We now construct a quantum circuit that realizes the unitary matrix \(\rU\) as described in Theorem~\ref{theorem:Grover-Rudolph}, such that
\[
\mathbb{P}_{\rho_0}(\mathsf{A} = k) = \int_{\frac{k}{2^n}}^{\frac{k+1}{2^n}} \varrho(x) \, dx, \quad \text{for each } k \in \mathbb{Z}_{2^3}.
\]
This is implemented in a \(3\)-qubit universal digital quantum computer \((\rA, \rho_0, \mathcal{N})\) within the algebra \(\mathbb{M}_{8}(\mathbb{C})\). Consequently, from Theorem~\ref{theorem:Grover-Rudolph}, we must compute \(2^3 - 1 = 7\) unitary matrices or, equivalently, determine \(7\) parameters.

\bigskip

\noindent \textbf{Step 1}: Computing \(\mathsf{U}_1\)
We start with the unitary matrix
\[
\mathsf{U}_1 := \rI_{2^{2}} \otimes \mathsf{R}(\theta),
\]
where
\[
\theta = \arccos \sqrt{\frac{\int_{0}^{1/2} \varrho(x) dx}{\int_{0}^{1} \varrho(x) dx}} = \arccos \sqrt{\int_{0}^{1/2} \varrho(x) dx} = \frac{\pi}{4}.
\]
Figure~\ref{Figura_1} illustrates the value of the integral \(\int_{0}^{1} \varrho(x) dx\) (shaded in blue) and \(\int_{0}^{1/2} \varrho(x) dx\) (shaded in dark blue).

\begin{figure}[h]
    \centering
    \includegraphics[width=0.45\textwidth]{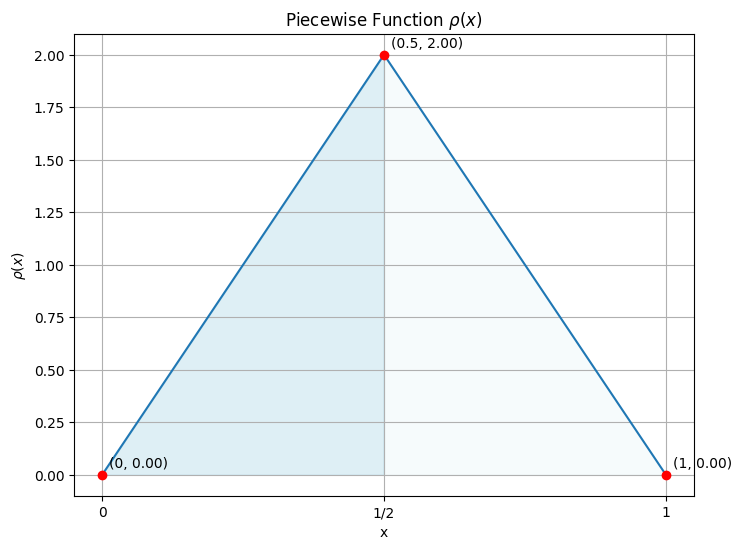}
    \caption{The integral \(\int_{0}^{1} \varrho(x) dx\) is colored in blue.}
    \label{Figura_1}
\end{figure}

\bigskip

\noindent \textbf{Step 2}: Computing \(\mathsf{U}_2\)
To construct
\[
\rU_2 := \prod_{z_2\in\mathbb{Z}_2}(\rI_{2} \otimes \rC\rC^{(1)}_{z_2} \mathsf{R}(\theta_{z_2})),
\]
we need the parameters
\[
\theta_0 = \arccos\sqrt{\frac{\int_{0}^{1/4} \varrho(x) \, dx}{\int_{0}^{1/2} \varrho(x) \, dx}} =  \frac{\pi}{3}, \quad
\theta_1 = \arccos\sqrt{\frac{\int_{1/2}^{3/4} \varrho(x) \, dx}{\int_{1/2}^{1} \varrho(x) \, dx}} = \frac{\pi}{6}.
\]

Figure~\ref{Figura_2} shows the decomposition of these integrals into their respective colored regions.

\begin{figure}[h]
    \centering
    \includegraphics[width=0.5\textwidth]{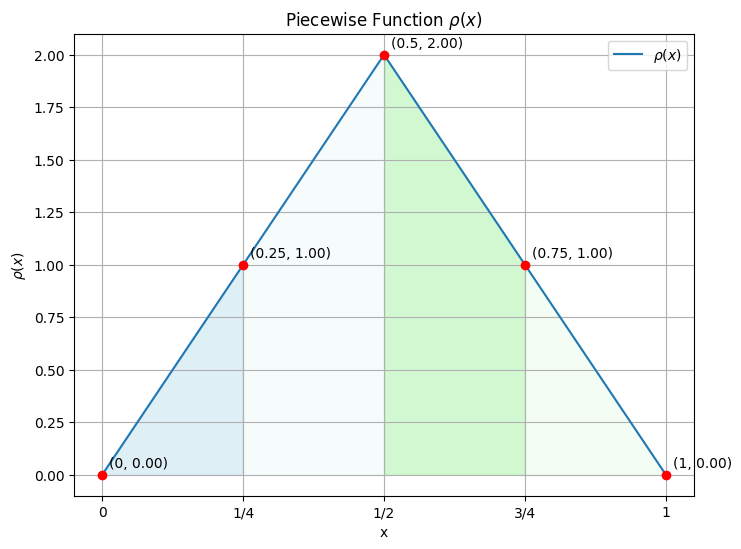}
    \caption{The integral \(\int_{0}^{1/2}\varrho(x) \, dx\) is in blue, and \(\int_{1/2}^{1} \varrho(x) \, dx\) is in green.}
    \label{Figura_2}
\end{figure}

\noindent \textbf{Step 3}: Computing \(\mathsf{U}_3\)
To finalize, we compute
\[
\mathsf{U}_3 = \prod_{(z_{1},z_2) \in \mathbb{Z}_{2}^2} \rC\rC^{(1)}_{z_{1}z_2} \mathsf{R}(\theta_{z_{1}z_2}),
\]
where the parameters are given by
\[
\theta_{00} = \frac{\pi}{3}, \quad
\theta_{11} = \frac{\pi}{6}, \quad
\theta_{01} = \arccos \frac{\sqrt{21}}{6}, \quad
\theta_{10} = \arccos \frac{\sqrt{15}}{6}.
\]
Figure~\ref{Figura_3} illustrates the integral regions required to compute these parameters. Thus, the complete quantum circuit \(\rU\) is given by
\[
\rU := \left(\prod_{(z_{1},z_2) \in \mathbb{Z}_{2}^2} \rC\rC^{(1)}_{z_{1}z_2} \mathsf{R}(\theta_{z_{1}z_2})\right)
 \left(\prod_{z_2\in\mathbb{Z}_2}(\rI_{2} \otimes \rC\rC^{(1)}_{z_2} \mathsf{R}(\theta_{z_2}))\right)
\left(\rI_{2^{2}} \otimes \mathsf{R}(\theta)\right).
\]

\begin{figure}[h]
    \centering
    \includegraphics[width=0.55\textwidth]{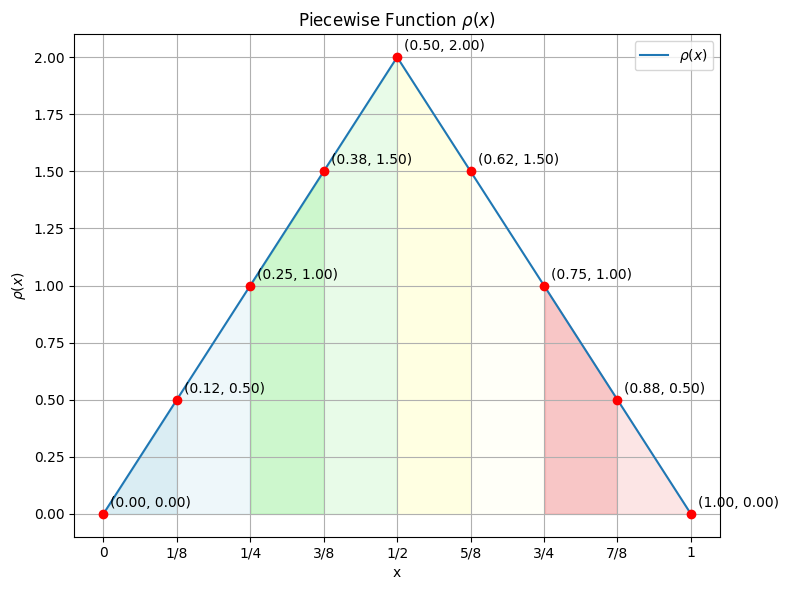}
    \caption{The integral areas used to determine parameters \(\theta_{00}, \theta_{01}, \theta_{10},\) and \(\theta_{11}\).}
    \label{Figura_3}
\end{figure}

\noindent\textbf{Quantum Circuit Implementation}: 
Figure~\ref{fig:quantikz-circuit} illustrates the corresponding quantum circuit. We simulate this quantum circuit using Qiskit's Aer module \texttt{qasm\_simulator}. Running \(2048\) experiments in a simulated \(3\)-qubit quantum computer, we obtain the probability distribution shown in Figure~\ref{Histogram}.

\begin{figure}[ht] % ht: opciones de ubicación (here, top)
    \centering
    \scalebox{0.8}{ % Escala para reducir el tamaño del circuito
    \begin{quantikz}[row sep=0.5cm, font=\scriptsize]
    \lstick{$z_0$} &&&\gate[3,style={fill=blue!10},label style=blue]{\rU_3}&\\
    \lstick{$z_1$}&&\gate[2,style={fill=green!10},label style=customgreen]{\rU_2}&&\\
    \lstick{$z_2$}&\gate[style={fill=red!10},label style=red]{\rU_1}&&&
    \end{quantikz} = 
    \begin{quantikz}[row sep=0.5cm, font=\scriptsize]
           &&&&& \gate[style={fill=blue!10},label style=blue]{\rR(\theta_{00})}& \gate[style={fill=blue!10},label style=blue]{\rR(\theta_{10})}&  \gate[style={fill=blue!10},label style=blue]{\rR(\theta_{01})} & \gate[style={fill=blue!10},label style=blue]{\rR(\theta_{11})} & \meter{} \\
           &\slice{}&&  \gate[style={fill=green!10},label style=customgreen]{\rR(\theta_0)} & \gate[style={fill=green!10},label style=customgreen]{\rR(\theta_1)}& \ocontrol{}&\control{}& \ocontrol{}&\control{}  & \meter{}\\
        & \gate[style={fill=red!10},label style=red]{\rR(\theta)}&&\octrl{-1}&\ctrl{-1}\slice{}  & \octrl{-2}&\octrl{-2}&\ctrl{-2}&\ctrl{-2} &  \meter{}
    \end{quantikz}
    }
    \caption{Quantum circuit \(\rU\) implementing Theorem~\ref{theorem:Grover-Rudolph}.} % Cambia la descripción aquí
    \label{fig:quantikz-circuit} % Etiqueta para referencias
\end{figure}
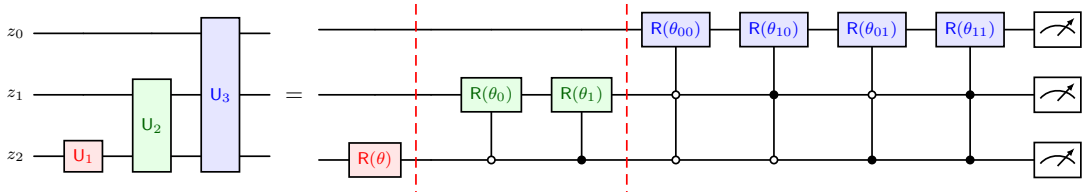

\begin{figure}[h]
    \centering
    \includegraphics[scale=0.5]{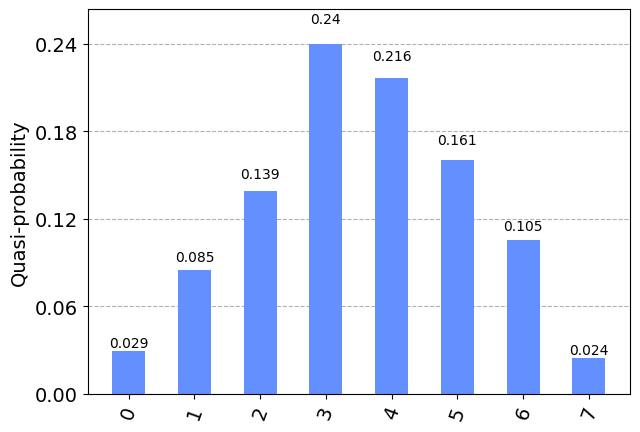}
    \caption{Outcome of the Grover-Rudolph algorithm on a simulated \(3\)-qubit quantum computer.}
    \label{Histogram}
\end{figure}

\bigskip

We can conclude that, despite the simplicity of the given example, it provides valuable insight into the correct functioning of the algorithm. From the comparison, it is evident that the computed values of the integrals of \(\varrho(x)\) are reasonably close to the true values. The maximum deviation occurs in the interval \([3/8, 1/2]\), where the computed integral is \(0.240\) compared to the true value of \(0.21875\), resulting in a difference of approximately \(0.02125\). The minimum deviation is observed in the interval \([1/2, 5/8]\), with a difference of \(0.00275\).  In general, these results indicate that the quantum algorithm \(\rU\) provides a reasonably accurate approximation of the true integrals of the probability density function \(\varrho(x)\) over the specified intervals. These results were obtained using the simulator provided by the Qiskit Python package.
%, as implemented in the code available at \textcolor{red}{[add GitHub repository link]}.

\bigskip

A natural next step would be to implement this circuit on a real quantum device, allowing for a more realistic assessment of its performance and potential limitations. Nevertheless, this example effectively illustrates the algorithm's behavior and facilitates a deeper understanding of its underlying principles.

\section{Conclusions and Final Remarks}\label{Conclusions}

In this work, we developed a novel mathematical framework for a universal digital quantum computer using the formalism of algebraic probability theory. This approach allowed us to rigorously define quantum circuits as finite sequences of elementary quantum gates and establish their role in implementing unitary transformations. Our framework was applied to the construction of quantum circuits that encode probability distributions, specifically in the context of the Grover-Rudolph algorithm.

\bigskip

Through our analysis, we demonstrated that every unitary transformation in \(\mathrm{U}(N)\) can be decomposed into a finite sequence of elementary quantum gates, leading to the concept of a \emph{universal dictionary}. This result guarantees that the set of elementary quantum gates \(\mathrm{QG}(N)\) forms a complete basis for quantum computation, a fundamental property for the realization of practical quantum algorithms.

\bigskip

A key outcome of our study was the explicit construction of a quantum circuit designed to implement a specific probability distribution via a sequence of controlled quantum gates and rotation matrices. This was verified numerically through a quantum simulation, confirming the theoretical predictions and demonstrating the effectiveness of our approach. Some directions for future research are the following:
\begin{itemize}
    \item \textbf{Generalization to Higher-Dimensional Systems:} Extending the algebraic probability framework to more complex quantum systems, such as continuous-variable quantum computers.
    \item \textbf{Optimization of Quantum Circuits:} Investigating circuit minimization techniques to reduce the number of elementary quantum gates required for a given unitary transformation.
    \item \textbf{Application to Quantum Machine Learning:} Exploring how the methods developed here can be used to encode and manipulate probability distributions relevant to quantum machine learning algorithms.
\end{itemize}

The results presented in this paper provide a rigorous mathematical foundation for quantum circuit synthesis within an algebraic probability framework. By bridging the gap between algebraic probability and quantum computing, our work offers a new perspective on quantum algorithm design. The Grover-Rudolph algorithm serves as a concrete example of how probability distributions can be precisely encoded in quantum circuits, opening doors for further advancements in quantum information processing.

\bigskip

As quantum technologies continue to evolve, the mathematical techniques introduced in this paper may serve as valuable tools for designing efficient quantum algorithms and expanding the capabilities of quantum computation. The theoretical insights developed in this work have the potential to contribute to the broader effort of enhancing the practicality and scalability of quantum computing in real-world applications

\bigskip

\noindent \textbf{Acknowledgements:} This research was funded by the grant number COMCUANTICA/007 from the Generalitat Valenciana, Spain and by the grant number INDI24/17 from the Universidad CEU Cardenal Herrera, Spain,

\begin{appendices}

\section{Proof of Theorem~\ref{Liouville-vonNeumann}}\label{Appendix_1}

To prove Theorem~\ref{Liouville-vonNeumann} we use the boundedness and self-adjointness of the commutator map induced by a fixed Hermitian matrix \(\rH\).

\begin{proposition}\label{bounded}
Let \(\rH \in \mathbb{M}_N(\mathbb{C})\) be an Hermitian matrix. Then the linear map 
\[
[\rH,\cdot] : (\mathbb{M}_N(\mathbb{C}), \|\cdot\|_{HS}) \longrightarrow (\mathbb{M}_N(\mathbb{C}), \|\cdot\|_{HS}),
\]
defined by \(\rA \mapsto [\rH,\rA]\), is bounded and self-adjoint.
\end{proposition}

\begin{proof}
Linearity is immediate from the definition of the commutator. Since \((\mathbb{M}_N(\mathbb{C}), \|\cdot\|_{HS})\) is a finite-dimensional Hilbert space, every linear map is bounded. To show that the map is self-adjoint, we must verify that
\[
\langle [\rH,\rA], \rB \rangle_{HS} = \langle \rA, [\rH,\rB] \rangle_{HS},
\]
for all \(\rA, \rB \in \mathbb{M}_N(\mathbb{C})\). Indeed, using the definition of the Hilbert--Schmidt inner product,
\[
\langle [\rH,\rA], \rB \rangle_{HS} = \operatorname{tr}\left([\rH,\rA]^\star \rB\right)
= \operatorname{tr}\left((\rA^\star \rH^\star - \rH^\star \rA^\star)\rB\right).
\]
Since \(\rH\) is Hermitian (\(\rH^\star=\rH\)), this becomes
\[
\operatorname{tr}\left(\rA^\star \rH \rB - \rH \rA^\star \rB\right)
= \operatorname{tr}\left(\rA^\star \rH \rB\right) - \operatorname{tr}\left(\rA^\star \rB \rH\right).
\]
By the cyclicity of the trace we have
\(
\operatorname{tr}\left(\rA^\star \rB \rH\right) = \operatorname{tr}\left(\rH\rA^\star \rB\right).
\)
Therefore, 
\[
\langle [\rH,\rA], \rB \rangle_{HS} = \operatorname{tr}\left(\rA^\star (\rH \rB - \rB\rH)\right)
= \langle \rA, [\rH,\rB] \rangle_{HS}.
\]
This shows that the map \( [\rH,\cdot] \) is self-adjoint.
\end{proof}

\bigskip

Our next result concerns the time evolution of the quantum state governed by the Liouville–von Neumann equation.

\begin{proof}[Proof of Theorem~\ref{Liouville-vonNeumann}]
By Proposition~\ref{bounded}, the differential equation \eqref{dynamics} is linear in the Hilbert space 
\(\bigl(\mathbb{M}_{N}(\mathbb{C}), \|\cdot\|_{HS}\bigr)\). Existence and uniqueness of its solutions 
are therefore guaranteed.  To verify that 
\(
\rho(t) \;=\; e^{-\,i\,t\,\rH} \,\rho_{0}\, e^{\,i\,t\,\rH}
\)
is indeed a solution, we differentiate with respect to \(t\):
\begin{align*}
\frac{d\rho(t)}{dt} & =
-\,i\,\rH\,e^{-\,i\,t\,\rH}\,\rho_{0}\,e^{\,i\,t\,\rH}
\;+\;
e^{-\,i\,t\,\rH}\,\rho_{0}\,\bigl(i\,\rH\bigr)\,e^{\,i\,t\,\rH} \\
& =
-\,i\,\bigl[\rH,\,e^{-\,i\,t\,\rH}\,\rho_{0}\,e^{\,i\,t\,\rH}\bigr]
\;=\;
-\,i\,[\rH,\,\rho(t)].
\end{align*}
Hence, \(\rho(t)\) satisfies the differential equation \eqref{dynamics}. 
Because \(\rho(t)\) is given in closed form by the exponential of a skew-Hermitian operator, 
and unitary conjugation preserves the rank of \(\rho_{0}\), 
all stated properties follow immediately.
\end{proof}

\section{Proof of Theorem~\ref{thm:universal_dictionary}}\label{Appendix_A}
  
  Along this appendix we assume that $N \in \mathbb{N}$ with $N\ge 2.$ Let $2 \le N < M$ two integers numbers, along this section we denote by  $\mathrm{Emb}(\mathrm{U}(N),\mathrm{U}(M))$ the set of injective group homomorphisms. We now prove a more general result, of which Theorem~\ref{thm:universal_dictionary} is a special case.

\begin{theorem}\label{general_universal_dictionary}
Let $N \ge 3.$ Given $\rU \in \mathrm{U}(N)$ there exists 
$m=\frac{N(N-1)}{2}$ and 
$$
(\rV_1,f_1),\ldots,(\rV_m,f_m) \in \mathrm{U}(2) \times 
\mathrm{Emb}(\mathrm{U}(2),\mathrm{U}(N)),
$$
such that 
$$
\rU = f_m(\rV_m)\cdots f_1(\rV_1).
$$
\end{theorem}

  Next. we introduce the following two maps 
  $\mathrm{up}_N,\mathrm{down}_N \in \mathrm{Emb}(\mathrm{U}(N-1),\mathrm{U}(N))$
  as follows. Given $\rU \in \mathrm{U}(N-1)$ we define
  $$
  \mathrm{up}_N(\rU) := \begin{bmatrix}
    \rU & \ket{0_{N-1}} \\
    \bra{0_{N-1}} & 1
  \end{bmatrix} \text{ and } \mathrm{down}_N(\rU) := \begin{bmatrix}
    1 & \bra{0_{N-1}} \\
    \ket{0_{N-1}} & \rU
  \end{bmatrix},
  $$
  where $\ket{0_{N-1}}$ represents the zero vector in $\mathbb{H}_{N-1}.$
  Observe that for $\ket{\Psi} \in \mathbb{H}_{N-1}$ and $\eta \in \mathbb{C}$ we have 
  $$
  \mathrm{up}_N(\rU)\begin{bmatrix}
    \ket{\Psi} \\
    \eta
  \end{bmatrix} = \begin{bmatrix}
    \rU \ket{\Psi} \\
    \eta
  \end{bmatrix} \text{ and } \mathrm{down}_N(\rU)\begin{bmatrix}
    \eta \\
    \ket{\Psi}
    \end{bmatrix} = \begin{bmatrix}
      \eta \\
      \rU \ket{\Psi}
    \end{bmatrix}.
  $$

  Finally, given $N\ge 3$ we introduce for $1 \le i < j \le N$ the map $K_{ij}^{(N)}\in \mathrm{Emb}(\mathrm{U}(2),\mathrm{U}(N))$ as follows
  $$
  K_{ij}^{(N)}(\rU) = K_{ij}^{(N)}\left(\begin{bmatrix}
    u_{11} & u_{12} \\
    u_{21} & u_{22}
  \end{bmatrix}\right)
  = 
  \left[
  \begin{array}{cccccccccccc}
   1 & 0 & \cdots & 0 & 0 & 0 & \cdots & 0 & 0 & 0 & \cdots & 0 \\
   0 & 1  &\cdots &0& 0 & 0  &\cdots & 0 & 0 & 0 & \cdots& 0 \\
  \vdots & \vdots &\ddots & \vdots & \vdots & \vdots & \ddots &\vdots& \vdots &\vdots& \ddots &\vdots \\
  0 & 0  &\cdots &1& 0 & 0  &\cdots & 0 & 0 & 0 & \cdots& 0 \\
   0 & 0  &\cdots &0& u_{11} & 0  &\cdots & 0 & u_{12} & 0 & \cdots& 0  \\
   0 & 0  &\cdots &0& 0 & 1  &\cdots & 0 & 0 & 0 & \cdots& 0  \\
  \vdots & \vdots &\ddots & \vdots & \vdots & \vdots & \ddots &\vdots& \vdots &\vdots& \ddots &\vdots \\
   0 & 0  &\cdots &0& 0 & 0  &\cdots & 1 & 0 & 0 & \cdots& 0 \\
   0 & 0  &\cdots &0& u_{21} & 0  &\cdots & 0 & u_{22} & 0 & \cdots& 0  \\
   0 & 0  &\cdots &0& 0 & 0  &\cdots & 0 & 0 & 1 & \cdots& 0  \\
  \vdots & \vdots &\ddots & \vdots & \vdots & \vdots & \ddots &\vdots& \vdots &\vdots& \ddots &\vdots \\
  0 & 0 & \cdots & 0 & 0 & 0 & \cdots & 0 & 0 & 0 & \cdots & 1
  \end{array}\right],
  $$
  here we have a modification of the identity matrix, where $u_{11}$ is in the $ii$-th position, $u_{22}$ is in the $jj$-th position,$u_{21}$ is in the $ji$-th position and $u_{12}$ is in the $ij$-th position. It is not difficult to see that $K_{ij}^{(N)} \in \mathrm{Emb}(\mathrm{U}(2),\mathrm{U}(N))$ holds. This matrix has the following property
  $$
  K_{ij}^{(N)}(\rU) = K_{ij}^{(N)}\left(\begin{bmatrix}
    u_{11} & u_{12} \\
    u_{21} & u_{22}
  \end{bmatrix}\right)\begin{bmatrix}
  \Psi_1 \\
    \vdots \\
    \Psi_{i-1} \\
    \Psi_i \\
    \Psi_{i+1} \\
    \vdots \\
    \Psi_{j-1} \\
    \Psi_j \\
    \Psi_{j+1} \\
    \vdots \\
    \Psi_N
  \end{bmatrix} = \begin{bmatrix}
    \Psi_1 \\
    \vdots \\
    \Psi_{i-1} \\
    u_{11}\Psi_i + u_{12}\Psi_j \\
    \Psi_{i+1} \\
    \vdots \\
    \Psi_{j-1} \\
    u_{21}\Psi_i + u_{22}\Psi_j \\
    \Psi_{j+1} \\
    \vdots \\
    \Psi_N  
  \end{bmatrix}.
  $$
  
  \bigskip
  
  To prove Theorem~\ref{general_universal_dictionary} we need to shown the following two results.
  
    \begin{lemma}\label{reduction}
      Let $\ket{\Psi} \in \mathbb{H}_N$ be a non-zero vector for some $N\ge 3.$ Then there exists 
      $$
      (\rV_1^{(N)},i_1^{(N)}),\ldots,(\rV_{N-1}^{N},i_{N-1}^{(N)}) \in \mathrm{U}(2)  \times \mathrm{Emb}(\mathrm{U}(2),\mathrm{U}(N))
      $$ 
      such that
      $$
      i_{N-1}^{(N)}(\rV_{N-1}^{(N)}) \cdots i_1^{(N)}(\rV_1^{(N)}) \ket{\Psi} =  \begin{bmatrix} \|\ket{\Psi}\| \\ 0 \\ \vdots  \\ 0 \end{bmatrix}
      $$
      \end{lemma}
      \begin{proof}
      We prove this lemma by induction on $N.$ First, we consider the case $N=3.$ Assume that $\ket{\Psi} = \begin{bmatrix} \Psi_1 \\ \Psi_2 \\ \Psi_3 \end{bmatrix}$ and suppose first that $\Psi_1 \neq 0.$ Then we take the matrix
      $$
      K_{12}^{(3)}\left( \begin{bmatrix}
        \frac{\overline{\Psi_1}}{\sqrt{ |\Psi_1|^2+|\Psi_2|^2}} & \frac{\overline{\Psi_2}}{\sqrt{ |\Psi_1|^2+|\Psi_2|^2}} \\
        \frac{\Psi_2}{\sqrt{ |\Psi_1|^2+|\Psi_2|^2}} & -\frac{\Psi_1}{\sqrt{ |\Psi_1|^2+|\Psi_2|^2}} \end{bmatrix} \right)  =\begin{bmatrix}
      \frac{\overline{\Psi_1}}{\sqrt{ |\Psi_1|^2+|\Psi_2|^2}} & \frac{\overline{\Psi_2}}{\sqrt{ |\Psi_1|^2+|\Psi_2|^2}} & 0 \\
      \frac{\Psi_2}{\sqrt{ |\Psi_1|^2+|\Psi_2|^2}} & -\frac{\Psi_1}{\sqrt{ |\Psi_1|^2+|\Psi_2|^2}} & 0 \\
      0 & 0 & 1
      \end{bmatrix}
      $$ 
      that satisfies
      $$
      \begin{bmatrix}
        \frac{\overline{\Psi_1}}{\sqrt{ |\Psi_1|^2+|\Psi_2|^2}} & \frac{\overline{\Psi_2}}{\sqrt{ |\Psi_1|^2+|\Psi_2|^2}} & 0 \\
        \frac{\Psi_2}{\sqrt{ |\Psi_1|^2+|\Psi_2|^2}} & -\frac{\Psi_1}{\sqrt{ |\Psi_1|^2+|\Psi_2|^2}} & 0 \\
        0 & 0 & 1
        \end{bmatrix} \begin{bmatrix} \Psi_1 \\ \Psi_2 \\ \Psi_3 \end{bmatrix} 
        = \begin{bmatrix} \sqrt{ |\Psi_1|^2+|\Psi_2|^2}\\ 0\\ \Psi_3\end{bmatrix}.
        .
      $$
      Now, we consider the matrix
      $$
      K_{13}^{(3)}\left( \begin{bmatrix}
        \frac{\sqrt{ |\Psi_1|^2+|\Psi_2|^2}}{\|\ket{\Psi}\|}  & \frac{\overline{\Psi_3}}{\|\ket{\Psi}\|} \\
        & \\
        \frac{\Psi_3}{\|\ket{\Psi}\|} & -\frac{\sqrt{ |\Psi_1|^2+|\Psi_2|^2}}{\|\ket{\Psi}\|}
        \end{bmatrix}\right) = \begin{bmatrix}
          \frac{\sqrt{ |\Psi_1|^2+|\Psi_2|^2}}{\|\ket{\Psi}\|} & 0  &\frac{\overline{\Psi_3}}{\|\ket{\Psi}\|} \\
        0 & 1 & 0 \\
        \frac{\Psi_3}{\|\ket{\Psi}\|} & 0 & -\frac{\sqrt{ |\Psi_1|^2+|\Psi_2|^2}}{\|\ket{\Psi}\|}
        \end{bmatrix},
      $$
      and hence 
      $$
      \begin{bmatrix}
        \frac{\sqrt{ |\Psi_1|^2+|\Psi_2|^2}}{\|\ket{\Psi}\|} & 0  &\frac{\overline{\Psi_3}}{\|\ket{\Psi}\|} \\
      0 & 1 & 0 \\
      \frac{\Psi_3}{\|\ket{\Psi}\|} & 0 & -\frac{\sqrt{ |\Psi_1|^2+|\Psi_2|^2}}{\|\ket{\Psi}\|}
      \end{bmatrix}\begin{bmatrix} \sqrt{ |\Psi_1|^2+|\Psi_2|^2}\\ 0\\ \Psi_3\end{bmatrix} = \begin{bmatrix} \|\ket{\Psi}\|\\ 0\\ 0\end{bmatrix}.
      $$
      Thus, there exists $(V_1^{(3)},i_1^{(3)}),(V_2^{(3)},i_2^{(3)}) \in \mathrm{U}(2) \times\mathrm{Emb}(\mathrm{U}(2),\mathrm{U}(3))$ such that 
      $$
      i_2^{(3)}(V_2^{(3)})i_1^{(3)}(V_1^{(3)})\ket{\Psi} = \begin{bmatrix} \|\ket{\Psi}\|\\ 0\\ 0\end{bmatrix}.
      $$
  The cases $\Psi_1 = 0,\Psi_2\neq 0$ and $\Psi_1 = 0,\Psi_2 =0,\Psi_3\neq 0$ are solved by using the product of matrices
  $$
  K_{12}^{(3)}\left( \begin{bmatrix} 0 & 1 \\ 1 & 0 \end{bmatrix}\right)K_{23}^{(3)}\left(\begin{bmatrix}
    \frac{\overline{\Psi_2}}{\sqrt{ |\Psi_2|^2+|\Psi_3|^2}} & \frac{\overline{\Psi_3}}{\sqrt{ |\Psi_2|^2+|\Psi_3|^2}} \\
    \frac{\Psi_3}{\sqrt{ |\Psi_2|^2+|\Psi_3|^2}} & -\frac{\Psi_2}{\sqrt{ |\Psi_2|^2+|\Psi_3|^2}} \end{bmatrix} \right),
  $$
  because
  \begin{align*}
  K_{12}^{(3)}\left( \begin{bmatrix} 0 & 1 \\ 1 & 0 \end{bmatrix}\right)K_{23}^{(3)}\left(\begin{bmatrix}
    \frac{\overline{\Psi_2}}{\sqrt{ |\Psi_2|^2+|\Psi_3|^2}} & \frac{\overline{\Psi_3}}{\sqrt{ |\Psi_2|^2+|\Psi_3|^2}} \\
    \frac{\Psi_3}{\sqrt{ |\Psi_2|^2+|\Psi_3|^2}} & -\frac{\Psi_2}{\sqrt{ |\Psi_2|^2+|\Psi_3|^2}} \end{bmatrix} \right)\begin{bmatrix} 0 \\ \Psi_2 \\ \Psi_3 \end{bmatrix} \\ 
      = K_{12}^{(3)}\left( \begin{bmatrix} 0 & 1 \\ 1 & 0 \end{bmatrix}\right) \begin{bmatrix} 0\\ \sqrt{ |\Psi_2|^2+|\Psi_3|^2}\\ 0 \end{bmatrix} = \begin{bmatrix} \|\ket{\Psi}\| \\ 0\\ 0 \end{bmatrix}.,
  \end{align*}
  
      In consequence the result follows for $N=3.$ Assume the result is true for $N-1$ and consider 
      $$
      \ket{\Psi} = \begin{bmatrix} \Psi_1 \\ \Psi_2 \\ \vdots \\ \Psi_{N} \end{bmatrix} = \begin{bmatrix}\ket{\Psi^{\prime}}\\ \Psi_N \end{bmatrix} \in \mathbb{H}_N, \text{ where } \ket{\Psi'} = \begin{bmatrix} \Psi_1 \\ \Psi_2 \\ \vdots \\ \Psi_{N-1} \end{bmatrix} \in \mathbb{H}_{N-1}
      $$ 
      be a non-zero vector. If $\ket{\Psi^{\prime}} \neq 0,$ we can apply the induction hypothesis to $\ket{\Psi^{\prime}}.$ This implies the existence of $(\rV_1^{(N-1)},i_1^{(N-1)}),\ldots,(\rV_{N-2}^{(N-1)},i_{N-2}^{(N-1)}) \in\mathrm{U}(2) \times\mathrm{Emb}(\mathrm{U}(2),\mathrm{U}(N-1)) $ such that 
      $$
      i_{N-2}^{(N-1)}(\rV_{N-2}^{(N-1)}) \cdots i_1^{(N-1)}(\rV_1^{(N-1)}) \ket{\Psi^{\prime}} =  \begin{bmatrix} \|\ket{\Psi^{\prime}}\| \\ 0 \\ \vdots  \\ 0 \end{bmatrix}.
      $$
    Hence, 
    \begin{align*}
    & \mathrm{up}_N(i_{N-2}^{(N-1)}(\rV_{N-2}^{(N-1)})) \cdots  \mathrm{up}_N(i_1^{(N-1)}(\rV_1^{(N-1)})) \ket{\Psi} \\ 
    & \\
    = & \mathrm{up}_N\left(i_{N-2}^{(N-1)}(\rV_{N-2}^{(N-1)}) \cdots  i_1^{(N-1)}(\rV_1^{(N-1)})\right)  \begin{bmatrix}\ket{\Psi^{\prime}}\\ \Psi_N \end{bmatrix}
    \\ 
    = & \begin{bmatrix} i_{N-2}^{(N-1)}(\rV_{N-2}^{(N-1)}) \cdots i_1^{(N-1)}(\rV_1^{(N-1)}) \ket{\Psi^{\prime}} \\ \Psi_N \end{bmatrix} = \begin{bmatrix} \|\ket{\Psi^{\prime}}\| \\ 0 \\ \vdots  \\ 0 \\ \Psi_N \end{bmatrix}.
    \end{align*}
    To conclude, we take the matrix 
    $$
  K_{1(N-1)}^{(N)}\left(
    \begin{bmatrix}
      \frac{\|\ket{\Psi^{\prime}}\|}{\|\ket{\Psi}\|} & \frac{\overline{\Psi_N}}{\|\ket{\Psi}\|} \\
      & \\
      \frac{\Psi_N}{\|\ket{\Psi}\|} & - \frac{\|\ket{\Psi^{\prime}}\|}{\|\ket{\Psi}\|}
    \end{bmatrix}
  \right)\begin{bmatrix} \|\ket{\Psi^{\prime}}\| \\ 0 \\ \vdots  \\ 0 \\ \Psi_N \end{bmatrix} = \begin{bmatrix} \|\ket{\Psi}\| \\ 0 \\ \vdots  \\ 0 \\ 0 \end{bmatrix},
    $$
  here we use that $\|\ket{\Psi}\|^2 =\|\ket{\Psi^{\prime}}\|^2 + |\Psi_N|^2.$ In consequence we have that 
  $$
  K_{1(N-1)}^{(N)}(V_{N-1}^{(N-1)})(\mathrm{up}_N \circ i_{N-2}^{(N-1)})(\rV_{N-2}^{(N-1)})) \cdots  (\mathrm{up}_N \circ i_1^{(N-1)})(\rV_1^{(N-1)})) \ket{\Psi} = \begin{bmatrix} \|\ket{\Psi}\| \\ 0 \\ \vdots  \\ 0 \\ 0 \end{bmatrix},
  $$
  where 
  $
  (V_{N-1}^{(N-1)},K_{1(N-1)}^{(N)}),(\rV_{N-2}^{(N-1)},\mathrm{up}_N \circ i_{N-2}^{(N-1)}),\ldots,(\rV_1^{(N-1)},\mathrm{up}_N \circ i_1^{(N-1)})$ are in  $\mathrm{U}(2) \times\mathrm{Emb}(\mathrm{U}(2),\mathrm{U}(N)) .$ Otherwise, if $\ket{\Psi^{\prime}} = 0,$ then $\Psi_N \neq 0,$ and we have
  \begin{align*}
    K_{12}^{(N)}\left( \begin{bmatrix}
      0 & 1 \\
     1 & 0 \end{bmatrix}
     \right) \cdots K_{N-2(N-1)}^{(N)}\left( \begin{bmatrix}
      0 & 1 \\
     1 & 0 \end{bmatrix}
     \right) K_{N(N-1)}^{(N)}\left( \begin{bmatrix}
     0 & \frac{\overline{\Psi_N}}{|\Psi_N|} \\
    \frac{\Psi_N}{|\Psi_N|} & 0 \end{bmatrix}
    \right)\begin{bmatrix} 0 \\ 0 \\ \vdots \\ 0 \\ \Psi_{N} \end{bmatrix} = \begin{bmatrix}  |\Psi_N| \\ 0 \\ \vdots \\ 0 \\ 0 \end{bmatrix}.
    \end{align*}
  This concludes the proof of the lemma.
  \end{proof}
  
  \bigskip
  
  \begin{proposition}\label{thm:universal_dictionary_extended}
    Assume $N \ge 3.$ Then for each $\rU \in \mathrm{U}(N)$ there exists 
    \begin{enumerate}
    \item[(a)] $(\rV_1^{(N)},i_1^{(N)}),\ldots,(\rV_{N-1}^{N},i_{N-1}^{(N)}) \in \mathrm{U}(2)  \times \mathrm{Emb}(\mathrm{U}(2),\mathrm{U}(N))$ and 
    \item[(b)] $(\rU_{N-1}^{(N)},\mathrm{down}_N) \in U(N-1)\times \mathrm{Emb}(\mathrm{U}(N-1),\mathrm{U}(N)) $
    \end{enumerate}
    such that 
    \begin{align*}
    i_{N-1}^{(N)}(\rV_{N-1}^{(N)}) \cdots i_1^{(N)}(\rV_1^{(N)}) \rU = \mathrm{down}_N(\rU_{N-1}^{N}).
    \end{align*}
    Furthermore, 
    \begin{align*}
     \rU =  \mathrm{down}_N(\rU_{N-1}^{(N)})i_{N-1}^{(N)}((\rV_{N-1}^{(N)})^{\star}) \cdots i_1^{(N)}((\rV_1^{(N)})^{\star}).
    \end{align*}
    \end{proposition}
  
    \begin{proof}
  Given $\rU \in U(N)$ we can write
    $$
    \rU = \begin{bmatrix}
      \ket{\Psi_1} & \ket{\Psi_2} & \cdots  & \ket{\Psi_N}
    \end{bmatrix},
    $$
    where $\{\ket{\Psi_1},\ket{\Psi_2},\ldots ,\ket{\Psi_N}\}$ is an orthonormal basis of $\mathbb{H}_N.$ From Lemma~\ref{reduction}
    there exists 
    $$
    (\rV_1^{(N)},i_1^{(N)}),\ldots,(\rV_{N-1}^{N},i_{N-1}^{(N)}) \in \mathrm{U}(2)  \times \mathrm{Emb}(\mathrm{U}(2),\mathrm{U}(N))
    $$ 
    such that
      $$
      i_{N-1}^{(N)}(\rV_{N-1}^{(N)}) \cdots i_1^{(N)}(\rV_1^{(N)})  \ket{\Psi_1} =  \begin{bmatrix} 1 \\ 0 \\ \vdots \\ 0 \end{bmatrix}.
      $$
    In consequence, for $\rZ:= i_{N-1}^{(N)}(\rV_{N-1}^{(N)}) \cdots i_1^{(N)}(\rV_1^{(N)})$ we have
    \begin{align*}
    \rZ \rU  & = \begin{bmatrix}
      \rZ \ket{\Psi_1} &  \rZ \ket{\Psi_2} & \cdots  & \rZ \ket{\Psi_N}
    \end{bmatrix} =
    \begin{bmatrix}
    1 & u_{12} & \cdots & u_{1N} \\
    0 & u_{22} & \cdots & u_{2N} \\
    \vdots & \vdots & \ddots & \vdots \\
    0 & u_{N2} & \cdots & u_{NN}
    \end{bmatrix}.
  \end{align*}
    Since, the columns of the matrix $\rZ \rU$ are also an orthonormal basis of $\mathbb{H}_N,$ we deduce that
    $u_{12}= \cdots =u_{1N} =0.$ Thus, we have that
    $$
    i_{N-1}^{(N)}(\rV_{N-1}^{(N)}) \cdots i_1^{(N)}(\rV_1^{(N)})  \rU = \mathrm{down}_N\left(\begin{bmatrix}
      u_{22} & \cdots & u_{2N} \\
       \vdots & \ddots & \vdots \\
      u_{N2} & \cdots & u_{NN}
      \end{bmatrix}\right).
    $$
    This proves the proposition.
  \end{proof}
  
  \bigskip
  
  Now, we are able to prove Theorem~\ref{general_universal_dictionary}.
  
  \bigskip
  
  \begin{proof}[Proof of Theorem~\ref{general_universal_dictionary}]
We proceed by induction on \(N\).

\paragraph{Base Case \(\boldsymbol{N=3}\).}
Let \(\rU \in \mathrm{U}(3)\). By Proposition~\ref{thm:universal_dictionary_extended}, 
there exist 
\[
(\rV_1^{(3)}, i_1^{(3)}), \quad
(\rV_2^{(3)}, i_2^{(3)}) 
\;\in\; 
\mathrm{U}(2) \,\times\, \mathrm{Emb}\!\bigl(\mathrm{U}(2), \mathrm{U}(3)\bigr)
\]
and 
\[
(\rU_{2}^{(3)}, \mathrm{down}_3) 
\;\in\; 
\mathrm{U}(2) \,\times\, \mathrm{Emb}\!\bigl(\mathrm{U}(2), \mathrm{U}(3)\bigr)
\]
such that
\[
i_{1}^{(3)}\!\bigl(\rV_{1}^{(3)}\bigr)\,
i_{2}^{(3)}\!\bigl(\rV_{2}^{(3)}\bigr)\,\rU 
\;=\; 
\mathrm{down}_3\!\bigl(\rU_{2}^{(3)}\bigr).
\]
Rearranging, we get
\[
\rU 
\;=\;
i_{2}^{(3)}\!\Bigl(\bigl(\rV_{2}^{(3)}\bigr)^{\star}\Bigr)\,
i_{1}^{(3)}\!\Bigl(\bigl(\rV_{1}^{(3)}\bigr)^{\star}\Bigr)\,
\mathrm{down}_3\!\bigl(\rU_{2}^{(3)}\bigr).
\]
Thus, the statement of the theorem holds for \(N = 3\).

\paragraph{Inductive Step.}
Assume the result is true for \(N - 1\). Let \(\rU \in \mathrm{U}(N)\). 
Applying Proposition~\ref{thm:universal_dictionary_extended} once again, 
we obtain 
\[
(\rV_1^{(N)}, i_1^{(N)}), \ldots, (\rV_{N-1}^{(N)}, i_{N-1}^{(N)}) 
\;\in\; 
\mathrm{U}(2)\,\times\, \mathrm{Emb}\!\bigl(\mathrm{U}(2), \mathrm{U}(N)\bigr)
\]
and 
\[
\bigl(\rU_{N-1}^{(N)}, \mathrm{down}_N\bigr) 
\;\in\; 
\mathrm{U}(N-1)\,\times\, \mathrm{Emb}\!\bigl(\mathrm{U}(N-1), \mathrm{U}(N)\bigr)
\]
such that
\begin{align}\label{decompositionA}
i_{N-1}^{(N)}\!\bigl(\rV_{N-1}^{(N)}\bigr)\,\cdots\,i_1^{(N)}\!\bigl(\rV_1^{(N)}\bigr)\,\rU 
\;=\; 
\mathrm{down}_N\!\bigl(\rU_{N-1}^{(N)}\bigr).
\end{align}
Since \(\rU_{N-1}^{(N)} \in \mathrm{U}(N-1)\), the induction hypothesis applies. 
Hence there exist 
\[
m \;=\; \frac{(N-1)\,(N-2)}{2}
\quad\text{and}\quad
(\rV_1,f_1), \ldots, (\rV_m,f_m) 
\;\in\;
\mathrm{U}(2)\,\times\, \mathrm{Emb}\!\bigl(\mathrm{U}(2), \mathrm{U}(N-1)\bigr)
\]
such that
\[
\rU_{N-1}^{(N)} 
\;=\; 
f_m(\rV_m)\,\cdots\,f_1(\rV_1).
\]
Therefore,
\[
\mathrm{down}_N\!\bigl(\rU_{N-1}^{(N)}\bigr)
\;=\;
(\mathrm{down}_N \circ f_m)(\rV_m)\,\cdots\,(\mathrm{down}_N \circ f_1)(\rV_1),
\]
where each pair 
\(\bigl(\rV_j,\mathrm{down}_N \circ f_j\bigr)\) 
belongs to 
\(\mathrm{U}(2)\,\times\, \mathrm{Emb}\!\bigl(\mathrm{U}(2), \mathrm{U}(N)\bigr)\). 
From \eqref{decompositionA}, it follows that
\[
\rU 
\;=\; \,
i_{1}^{(N)}\!\Bigl(\bigl(\rV_{1}^{(N)}\bigr)^{\star}\Bigr)\,\cdots\,i_{N-1}^{(N)}\!\Bigl(\bigl(\rV_{N-1}^{(N)}\bigr)^{\star}\Bigr)
(\mathrm{down}_N \circ f_m)(\rV_m)\,\cdots\,
(\mathrm{down}_N \circ f_1)(\rV_1).
\]
Counting the total number of \(\mathrm{U}(2)\) factors, we see that
\[
(N - 1) 
\;+\; 
\frac{(N - 2)\,(N - 1)}{2}
\;=\;
\frac{N\,(N - 1)}{2}
\]
is precisely the total number of 2\(\times\)2 unitaries required. 
This completes the proof.
\end{proof}

  \end{appendices}

%%===========================================================================================%%
%% If you are submitting to one of the Nature Portfolio journals, using the eJP submission   %%
%% system, please include the references within the manuscript file itself. You may do this  %%
%% by copying the reference list from your .bbl file, paste it into the main manuscript .tex %%
%% file, and delete the associated \verb+\bibliography+ commands.                            %%
%%===========================================================================================%%
\bibliography{bbl_general}% common bib file
%\bibliographystyle{plain}

%% if required, the content of .bbl file can be included here once bbl is generated
%%\input sn-article.bbl
\end{document}